\documentclass[conference,compsoc]{IEEEtran}
\IEEEoverridecommandlockouts
\usepackage{cite}
\usepackage{amsmath,amssymb,amsfonts}
\usepackage{amsthm}
\usepackage{graphicx}
\usepackage{textcomp}
\usepackage[table]{xcolor}
\usepackage{acronym}
\usepackage{enumitem}
\usepackage{algpseudocode}
\usepackage{graphicx}
\usepackage{tabularx}
\usepackage{booktabs}
\usepackage{makecell}
\usepackage{multirow}
\usepackage{wasysym}
\usepackage{float}
\usepackage{placeins}
\usepackage{algorithm}
\usepackage{hyperref}
\usepackage{xurl}

\usepackage[most]{tcolorbox}
\def\BibTeX{{\rm B\kern-.05em{\sc i\kern-.025em b}\kern-.08em
    T\kern-.1667em\lower.7ex\hbox{E}\kern-.125emX}}

\acrodef{FL}{Federated Learning}
\acrodef{SecAgg}[SecAgg]{Secure Aggregation}
\acrodef{SS}{Secret Sharing}
\acrodef{VSS}{Verifiable Secret Sharing}
\acrodef{ZKP}{Zero-Knowledge Proof}
\acrodef{NIZK}{Non-Interactive Zero-Knowledge}
\acrodef{zkSNARK}[zk-SNARK]{Zero-Knowledge Succinct Non-Interactive Argument of Knowledge}
\acrodef{HE}{Homomorphic Encryption}
\acrodef{TEE}{Trusted Execution Environment}
\acrodef{RA}{Remote Attestation}
\acrodef{DP}{Differential Privacy}
\acrodef{IVC}{Incrementally Verifiable Computation}
\acrodef{MPC}{Secure Multi-Party Computation}
\acrodef{ML}{Machine Learning}
\acrodef{GDPR}{General Data Protection Regulation}
\acrodef{PoC}{Proof-of-Concept}
\acrodef{EDPS}{European Data Protection Supervisor}
\acrodef{AEPD}{Spanish Data Protection Agency}

\acrodefplural{ZKP}[ZKPs]{Zero-Knowledge Proofs}
\acrodefplural{NIZK}[NIZKs]{Non-Interactive Zero-Knowledge Proofs}
\acrodefplural{zkSNARK}[zk-SNARKs]{Zero-Knowledge Succinct Non-Interactive Arguments of Knowledge}

\newcommand{\client}[1]{{{\mathsf{C}}^{(#1)}}}
\newcommand{\agg}[1]{{{\mathsf{A}}^{(#1)}}}
\newtheorem{definition}{Definition}
\newtheorem{theorem}{Theorem}
\newtheorem{lemma}{Lemma}
\newtheorem{corollary}{Corollary}
\newcommand{\Framework}{\mathfrak{F}}
\newcommand{\FLObj}{\mathsf{Obj}^{\mathsf{FL}}}
\newcommand{\CryptoObj}{\mathsf{Obj}^{\mathsf{Crypto}}}
\newcommand{\Params}{\mathsf{params}}
\newcommand{\Aux}{\mathsf{Aux}}
\newcommand{\Commit}{\mathsf{Commit}}

\newcommand{\Prove}{\mathsf{Prove}}
\newcommand{\Verify}{\mathsf{Verify}}
\newcommand{\LocalTrain}{\mathsf{LocalTrain}}
\newcommand{\Aggregate}{\mathsf{Aggregate}}
\newcommand{\Update}{\mathsf{Update}}
\newcommand{\Encode}{\mathsf{Encode}}

\newcommand{\ppnizk}{\mathsf{pp}_{\Pi}}

\algrenewcommand\algorithmicrequire{\textbf{Input:}}
\algrenewcommand\algorithmicensure{\textbf{Output:}}

\newcommand{\proto}{\emph{PoCoFL}}

\newtcolorbox{question}[1]{%
  enhanced,
  colback=gray!8,                 
  colframe=gray!70!black,         
  colbacktitle=gray!70!black,     
  coltitle=white,                 
  fonttitle=\sffamily\bfseries\small,
  fontupper=\small,               
  title={Compliance Question #1},
  boxrule=0.5pt, arc=1mm,
  left=4pt, right=4pt, top=2pt, bottom=2pt,
  toptitle=1pt, bottomtitle=1pt,
  before skip=6pt, after skip=6pt}

\newcommand{\GameBox}[2]{%
  \begingroup
  \setlength{\tabcolsep}{5pt}%
  \renewcommand{\arraystretch}{1.05}%
  \begin{tabularx}{\linewidth}{|>{\raggedright\arraybackslash}X|}
    \hline
    \textbf{#1}\\
    \hline
    #2\\
    \hline
  \end{tabularx}%
  \endgroup
}

\makeatletter
\newcommand{\CQbox}[1]{%
  \if@noskipsec\leavevmode\par\fi
  \par\smallskip\noindent
  \fbox{\parbox{\dimexpr\linewidth-2\fboxsep-2\fboxrule\relax}{#1}}%
  \par\smallskip}
\makeatother

\begin{document}

\title{PoCoFL: A Framework for POlicy-COmpliant Federated Learning}
\author{
\IEEEauthorblockN{
Dominik Roy George,
Varesh Mishra,
and Aysajan Abidin
}
\IEEEauthorblockA{\textit{COSIC, KU Leuven}, Belgium\\
dgeorge@esat.kuleuven.be, varesh.mishra@esat.kuleuven.be, aysajan@kuleuven.be}
}

\maketitle

\begin{abstract}
Federated Learning (FL) is a privacy-oriented learning paradigm that enables collaborative model training while maintaining training data locally on participating clients. However, FL does not guarantee that clients submit policy-compliant contributions or that aggregators process admitted contributions correctly. Update-magnitude bounds alone are insufficient to ensure compliance with learning-specific requirements involving private data, intermediate models, or retained state. These requirements differ across FL settings, motivating the need for a unified framework to specify and verify client and aggregation policies. This work introduces PoCoFL, a policy-compliant federated learning framework that separates three aspects: (i) FL type, (ii) policy semantics, and (iii) cryptographic realisation. We provide a formalisation that captures client and aggregation requirements as policy-dependent relations, defining update admissibility, aggregation correctness, and policy-relative privacy. These relations enable verification of selected properties of intermediate objects, such as local-model loss, continual-learning state, and personalised optimisation conditions, in each round. Clients demonstrate compliance of their contributions using compatible cryptographic mechanisms, while aggregators prove that the recorded set of admitted contributions was processed according to the selected aggregation policy. We demonstrate PoCoFL through four formal instantiations and corresponding proof-of-concept realisations: (i) vanilla, (ii) continual, (iii) personalised, and (iv) threshold-encrypted federated learning. We evaluate how policy enforcement affects the learning objectives of vanilla, personalised, and continual FL, and we quantify the computational and communication overhead of policy verification across all four instantiations. Overall, these results demonstrate that PoCoFL can capture complex policy representations across the evaluated FL settings.
\end{abstract}


\section{Introduction}\label{sec:introduction}
\ac{FL} lets several institutions train a shared model while keeping their sensitive training data on-site~\cite{fedavgmcmahan17a,kairouz2021advances}. Medical institutions, for instance, use it to develop clinical models across geographically distributed sites~\cite{rieke2020future,dayan2021covid,pati2022federated}. Their data are protected by the \ac{GDPR} in the European Union and by the HIPAA Privacy Rule in the United States~\cite{gdpr2016,hhsHipaaPrivacyRule}. By keeping patient records at their source and exchanging only model updates, \ac{FL} reduces the personal data transferred between institutions. The \ac{EDPS}, the \ac{AEPD}, and the ICO acknowledge this benefit for data minimisation~\cite{edps2025fl,ico2023pets}. Regulation, however, asks for more than reduced data sharing. Institutions must also demonstrate how personal data are processed and how the shared model is developed. The \ac{GDPR} obliges controllers to demonstrate compliance with its principles~\cite[Art.~5(2)]{gdpr2016}. The EU AI Act classifies AI systems that are medical devices, or safety components of medical devices, as high-risk when the device requires third-party conformity assessment~\cite[Art.~6(1)]{aiact2024}. Such systems must meet requirements on training-data governance, automatic logging over their lifecycle, and accuracy, robustness, and cybersecurity, including protection against data and model poisoning~\cite[Arts.~10, 12, 15]{aiact2024}. The UK AI Cyber Security Code of Practice and the OECD AI Principles similarly call for audit trails of models and datasets and for traceability across the AI lifecycle~\cite{dsit2025aicop,oecd2024ai}. For medical models updated after deployment, regulatory guidance asks that changes be planned and tested in advance~\cite{pccp2023principles}. These requirements call for technical evidence of how distributed training and model updates were carried out. This evidence should be available before the high-risk requirements of the AI Act take effect on 2 December 2027, and on 2 August 2028 for regulated products such as medical devices~\cite{aiomnibus2026}.

Producing this evidence in \ac{FL} is hard because no party observes the whole process. The server that combines contributions cannot access the training data and therefore cannot assess its quality~\cite{edps2025fl}. The institutions cannot observe how the server combines their contributions~\cite{wang2025zkfl}. Privacy protections may restrict this view further. A federation therefore agrees on rules that no single member can verify. We call such an agreed rule a \emph{policy}. A policy turns a regulatory requirement into a concrete condition that training or aggregation must satisfy.

\noindent\textbf{Illustrative example.}
Consider a consortium of hospitals that trains a diagnostic model under agreed policies. Verifying these policies raises four questions.

\textbf{(Q1)}~Does each contribution meet the agreed quality level on the hospital's own data?

\textbf{(Q2)}~When the model learns new tasks, does each contribution stay within the agreed limit on performance loss for earlier tasks?

\textbf{(Q3)}~When a hospital adapts the model to its own patients, is its shared contribution still computed from the agreed personalisation objective?

\textbf{(Q4)}~When contributions must stay hidden from each server, is each contribution well formed and is the aggregate correct?
Each question refers to private data, intermediate models, or state retained across rounds. This paper investigates how a federation can demonstrate that every contribution and aggregation step follows the agreed policies while keeping this information private.

\noindent\textbf{Limitations of existing approaches.}
Byzantine-robust aggregation filters contributions that appear statistically unusual~\cite{blanchard2017krum,cao2021fltrust}. It relies on an honest server that sees all contributions, often assumes similar client data~\cite{karimireddy2022bucketing}, and produces no verifiable evidence. Verifiable input validation lets a client prove that a hidden update satisfies a public predicate, typically a norm bound~\cite{chowdhury2022eiffel,lycklama2023rofl,bell2023acorn}. Such a bound is a useful policy, but it cannot address private data, intermediate models, or retained state, and harmful updates can stay within it~\cite{lycklama2023rofl,wang2020attacktails}. Proofs of complete training cover every step, but the prover pays for every step~\cite{garg2023zkpot,abbaszadeh2024kaizen}. Moreover, existing constructions bind the verified conditions, learning workflow, and cryptographic techniques to one protocol~\cite{rizzini2026sok}. A different policy then requires a new construction and security analysis.

\noindent\textbf{Our approach.}
We present \proto, a generic framework for policy-compliant \ac{FL}. Policies specify conditions on data, models, contributions, and retained state, and cryptographic commitments bind proofs to these objects without revealing them. Clients prove that their contributions satisfy the client policy. Aggregators prove that they combined the recorded set of admitted contributions according to the aggregation policy. \proto\ separates the \ac{FL} execution context, the policies, and the cryptographic mechanisms that verify them. This separation gives ML researchers and cryptographers a common interface for developing policies such as those behind (Q1) to (Q4). The framework provides update admissibility, aggregation correctness with respect to the recorded set of admitted contributions, and policy-relative privacy.

\noindent\textbf{Scope.}
\proto\ is not a Byzantine-robust aggregation rule. It makes the rules of a federation verifiable but does not choose them. A norm bound is one such rule, besides the loss, retention, and optimality conditions we instantiate. Compliance is relative to the selected policy, and a proof alone does not make a deployment lawful.

\noindent\textbf{Contributions.}
\begin{itemize}[leftmargin=*,labelindent=0pt]
  \item We derive five requirements for policy-compliant and verifiable \ac{FL} from deployment practice and from what regulation expects a federation to demonstrate, and we map each expectation to the technical evidence that addresses it.
  \item We formalise \proto\ and prove update admissibility, aggregation correctness, and policy-relative privacy under standard cryptographic assumptions.
  \item We instantiate and implement \proto\ for vanilla, continual, personalised, and threshold-\ac{HE} \ac{FL}. We also evaluate how each policy affects its learning objective and what its evidence costs.
\end{itemize}

\textbf{Roadmap.}
Sec.~\ref{sec:model} presents the system model, requirements, and threat model. Sec.~\ref{sec:framework} introduces the framework. Sec.~\ref{sec:security-analysis} states the security guarantees, with proofs in Appendix~\ref{app:security-proofs}. Sec.~\ref{sec:Instantiation} presents the instantiations and Sec.~\ref{sec:evaluation} their evaluation. Secs.~\ref{sec:discussion} and~\ref{sec:related-work} discuss the framework and related work. Sec.~\ref{sec:conclusion} concludes the paper.

\section{System Model, Requirements, and Threat Model}\label{sec:model}
This section establishes the setting in which the framework operates. We first describe the participating entities and \ac{FL} settings. We then derive the policy-compliance requirements from practical \ac{FL} deployments and the regulatory expectations of Sec.~\ref{sec:introduction}. Finally, we define the adversarial model and the security properties targeted by the framework.
\subsection{System Model}\label{sec:system-model}
In our system model, we consider the standard entities involved in an \ac{FL} setting, namely clients and aggregators. Depending on the selected \ac{FL} type, a participant may realise one of these roles or both. A client holds a private local dataset and uses it to train a local model without directly disclosing the underlying data. An aggregator receives the resulting client contributions, such as local model parameters or updates, and combines them according to the prescribed aggregation rule. Whenever a client or an aggregator verifies a proof, it additionally acts as a verifier. The verifier is therefore a logical role and does not require a separate trusted entity. We assume that communication between participating entities is established over authenticated and confidential channels.

\emph{Federation agreement and policies.}
The participants form a federation, such as the hospital consortium of Sec.~\ref{sec:introduction}. Before training starts, they fix a federation agreement that sets the \ac{FL} type, the public parameters, and the policies. The \ac{FL} type fixes the execution structure, namely the participating roles, communication topology, retained state, and aggregation rules. A policy is an agreed rule on a client-side or an aggregation-side action within this structure. It specifies the conditions under which the action is considered compliant and may combine several conditions over public inputs, private inputs, and retained state. In regulatory terms, the policies are the machine-checkable part of the documented training process.

\emph{Supported \ac{FL} types.}
In \emph{single server-assisted} types, a logical aggregation service coordinates the training rounds, receives the client contributions, and produces the state used in the next round. In classic \ac{FL}~\cite{fedavgmcmahan17a, yangFL2019}, which we refer to as \(\mathsf{VanillaFL}\), the clients train one shared model, for example a common model on X-ray images. \(\mathsf{PersonalizedFL}\)~\cite{dinh2020pfedme} adds client-specific models or personalisation state. It is relevant when hospitals serve different populations or use distinct imaging protocols~\cite{jiangTMI}, since a model can perform worse at a new site than at the sites that provided its training data~\cite{zech2018variable}. In \(\mathsf{ContinualFL}\)~\cite{Kim2024cl, ewc-paper}, clients learn from a sequence of tasks, for example to extend a diagnostic model to new diseases while limiting forgetting~\cite{pmlr-v139-yoon21b}, and retain state across tasks or rounds. Policies for these types may govern the admissibility of contributions, the relation between private personalised state and shared state, and the retained state.

In \emph{decentralised and multi-aggregator} types, participants may exchange contributions with their peers and act as both clients and aggregators, or several parties may jointly compute or authorise the aggregate. A realisation may also prevent a single party from recovering protected contributions, since shared updates can reveal the data they were computed from~\cite{zhu2019dlg}. Policies for these types may additionally govern participation, quorum, and the correctness of the joint computation. The types above are representative rather than exhaustive. A further setting is supported by defining its roles, round structure, policy-relevant state, communication pattern, and aggregation semantics, as formalised in Sec.~\ref{sec:framework}.

\emph{Verifiable evidence.}
Keeping the local data at the clients prevents its direct disclosure, but does not by itself guarantee that an action is policy-compliant. The participant performing an action therefore attaches verifiable evidence showing that the corresponding policy is satisfied. The designated verifier checks this evidence when processing the protocol message and accepts the action only if the required policy conditions hold. The evidence is non-interactive. A party that holds the recorded public messages of a round can therefore repeat the checks later. We refer to such a party as an \emph{auditor}. It may be an internal auditor, a notified body, or a supervisory authority, and it needs no access to private data, models, or state.

\begin{table}[t]
\centering
\caption{Regulatory expectations for a federation and the requirements (Sec.~2.2) and properties (Sec.~2.3) of \proto\ that address them. The mapping is not a legal assessment.}
\label{tab:regulatory-mapping}
\footnotesize
\setlength{\tabcolsep}{3pt}
\renewcommand{\arraystretch}{1.3}
\setlength{\aboverulesep}{0pt}
\setlength{\belowrulesep}{0pt}
\rowcolors{2}{white}{gray!12}
\begin{tabularx}{\columnwidth}{>{\raggedright\arraybackslash}X cc}
\toprule
\rowcolor{white}
\textbf{Regulatory expectation} & \textbf{Req.} & \textbf{Prop.} \\
\midrule
Demonstrable compliance and traceability of the training process
  \newline {\scriptsize (\ac{GDPR} Arts.~5(2), 24(1)~\cite{gdpr2016}, UK Code of Practice~\cite{dsit2025aicop}, OECD~\cite{oecd2024ai})}
  & R2, R3 & P1, P2 \\
Records over the system lifetime
  \newline {\scriptsize (AI Act Art.~12(1)~\cite{aiact2024})}
  & R2, R3 & P1, P2 \\
Origin and quality of training data
  \newline {\scriptsize (AI Act Art.~10(2)--(3)~\cite{aiact2024}, EDPS--AEPD~\cite{edps2025fl})}
  & R2 & P1 \\
Consistent performance and controlled change of a learning model
  \newline {\scriptsize (AI Act Arts.~15(1), 43(4)~\cite{aiact2024}, FDA/HC/MHRA~\cite{pccp2023principles})}
  & R2 & P1 \\
Suitability for the setting of use
  \newline {\scriptsize (AI Act Art.~10(4)~\cite{aiact2024})}
  & R1, R2 & P1 \\
Data minimisation and security of processing, beyond \ac{FL} alone
  \newline {\scriptsize (\ac{GDPR} Arts.~5(1)(c),(f), 32~\cite{gdpr2016}, ICO~\cite{ico2023pets}, EDPS--AEPD~\cite{edps2025fl})}
  & R4 & P3 \\
Resilience to manipulation, incl.\ data and model poisoning
  \newline {\scriptsize (AI Act Art.~15(5)~\cite{aiact2024}, UK Code of Practice~\cite{dsit2025aicop})}
  & R2, R3 & P1, P2 \\
State-of-the-art security measures
  \newline {\scriptsize (\ac{GDPR} Arts.~25(1), 32(1)~\cite{gdpr2016})}
  & R5 & NA \\
\bottomrule
\end{tabularx}
\end{table}
\subsection{Policy-Compliance Requirements}\label{sec:requirements}
The requirements follow from two sources, the \ac{FL} types discussed above and the regulatory expectations in Table~\ref{tab:regulatory-mapping}. The provisions state what a federation must demonstrate but leave the technical means open. They concern both sides of the protocol, and the table maps each expectation to the requirements below and to the properties that formalise them.

\textbf{(R1) Type-aware policy configuration.}
The \ac{FL} types differ in their roles, topology, retained state, and aggregation semantics, so no single rule fits every federation. The AI Act also expects the data to reflect the specific setting of use (Table~\ref{tab:regulatory-mapping}), which a federation may address by personalising the model to each hospital's population. The framework should therefore support different \ac{FL} types and network topologies and let each federation configure the policies for its type. It should also support reusable policies that apply across types, such as verifying that a score is computed correctly and satisfies a given threshold.


\textbf{(R2) Verifiable client-policy compliance.}
Every client contribution may influence an aggregation result or a participant’s model state. Therefore, each contribution should be accompanied by verifiable evidence that it complies with the selected client policy and is consistent with its context. The policy may govern, for example, participant and round identifiers, committed private data, loss conditions, or continuity with retained client state, to name a few. A federation can then demonstrate data governance and controlled model change for every contribution. This requirement is formalised by Property~\textbf{(P1)}, \emph{update admissibility}, in Sec.~\ref{sec:threat-model}.

\textbf{(R3) Verifiable aggregation-policy compliance.}
Every aggregation result should be accompanied by verifiable evidence showing that the aggregator followed the policy applicable to the selected \ac{FL} type. The policy may govern the set of compliant updates (recorded admitted set), aggregation values and aggregation rule, to name a few. Accountability thus applies to the party that assembles the model as much as to the contributing parties. This requirement is formalised by Property~\textbf{(P2)}, \emph{aggregation correctness relative to the recorded set}, in Sec.~\ref{sec:threat-model}.

\textbf{(R4) Private policy verification.}
A policy may depend on private assets, such as a client's local data, retained state, commitment openings, randomness, or intermediate values. The designated verifier should be able to establish policy compliance without learning these assets beyond the public statement and the compliance outcome. Otherwise, the evidence would undo the \emph{data minimisation} that motivates \ac{FL}. This requirement is formalised by Property~\textbf{(P3)}, \emph{policy-relative privacy}, in Sec.~\ref{sec:threat-model}. Confidentiality of individual contributions during aggregation is separate and may be provided by an auxiliary mechanism selected by the realisation.

\textbf{(R5) Cryptographic modularity and agility.}
The framework should separate policy logic from its cryptographic realisation. Cryptographic mechanisms should be defined through clear interfaces and may be replaced by compatible alternatives that provide the required functionality and security guarantees. Since some mechanisms depend on each other, replacing one component may also require replacing related components. Security obligations refer to the state of the art~\cite[Arts.~25(1), 32(1)]{gdpr2016}, which changes during the lifetime of a deployed system. The framework should therefore support cryptographic agility, including post-quantum alternatives where appropriate. Such changes should preserve the security properties of the framework and the semantics of its policies.

\subsection{Threat Model}\label{sec:threat-model}
The adversary in our model is a participant that deviates from the agreed policy. The reason for the deviation does not matter. The participant may be compromised, faulty, negligent, or malicious. We consider probabilistic polynomial-time ($\mathsf{PPT}$) adversaries acting either as malicious clients or as a malicious aggregation server.

\textbf{Malicious clients ($\mathcal{A}_1$).}
The adversary $\mathcal{A}_1$ controls the private inputs, state, randomness, and outgoing messages of one or more malicious clients. Its objective is to have a malformed or policy-violating contribution admitted for aggregation. Depending on the selected \ac{FL} type and client policy, a bad contribution may, for example, be inconsistent with the active model or round, fail a loss condition on the committed data, or violate a continual-learning condition. $\mathcal{A}_1$ does not obtain the datasets or private states of honest clients.

\textbf{Malicious server ($\mathcal{A}_2$).}
The adversary $\mathcal{A}_2$ controls the aggregation server and may output an arbitrary admitted-message record, aggregation result, and aggregation proof. It may attempt to replay or modify client messages, duplicate or substitute recorded contributions, change aggregation weights, or compute a result inconsistent with the prescribed aggregation policy. It receives the public client messages but does not obtain the private witnesses of honest clients.

Unless stated otherwise, we consider $\mathcal{A}_1$ and $\mathcal{A}_2$ separately and exclude collusion between malicious clients and the server. The public parameters are honestly generated, and protocol messages are bound to their claimed senders by the authenticated-channel assumption.

\noindent\textbf{Scope of the guarantees.}
Two clarifications delimit what a policy proof establishes. First, the guarantees are relative to the selected policy. \proto\ is not a Byzantine-robust aggregation rule and does not aim to bound the influence of an arbitrary fraction of clients. A contribution that satisfies the policy is admitted, even if it does not improve the model. How well a policy protects a learning objective is a property of that policy, and we study it empirically in Sec.~\ref{sec:evaluation}. Second, a client proves a property of the data it committed to. The proof does not show that this data is genuine or representative. Regulation assigns this question to data governance~\cite[Art.~10]{aiact2024}. \proto\ is broad enough to provide these guarantees through auxiliary cryptographic mechanisms, e.g., Digital Signatures. A federation can support it by registering dataset commitments when a client joins, so that later proofs refer to the registered data. 
We also emphasize that Denial-of-Service attacks and omission before an update becomes part of the public admitted-message record are outside the scope of this work. Thus, aggregation correctness is relative to the record $\mathsf{Adm}_t$ used for the round. It does not prove that every submitted update was placed in that record. Auxiliary mechanisms such as a public append-only log or another auditable submission mechanism can provide that additional guarantee.

\subsubsection{Security Properties}
Under the adversarial model above, the framework targets three guarantees. 

\textbf{(P1) Update admissibility.}
A client contribution is admitted only if its sender and public context match the active client, model, round, policy, and parameters. Its aggregation value is well formed and its statement belongs to the language $\mathcal{L}_{\phi}$ induced by the applicable client policy.

\textbf{(P2) Aggregation correctness relative to the recorded set.}
An aggregation result is accepted only if two conditions hold. The recorded set (we denote it by $\mathsf{Adm}_t$) contains pairwise-distinct admissible client messages. The output is obtained by extracting every recorded contribution exactly once, applying the prescribed aggregation rule and weights, and applying the prescribed model-update rule. This prevents omission, duplication, substitution, modification, or reweighting within the recorded set.

\textbf{(P3) Policy-relative privacy.}
For any two private client inputs that satisfy the same policy and induce the same declared public information and compliance outcome, the resulting verification transcripts are computationally indistinguishable to any $\mathsf{PPT}$ verifier. Hence verification reveals no information about the private policy witness beyond the public protocol inputs and the compliance outcome.

\section{Framework Description}\label{sec:framework}
We now present \proto. A federation agreement fixes a learning setting, a set of rules, and the means to check them. \proto\ mirrors this structure with three groups of components, namely the \ac{FL} type ($\mathsf{Type}$), the governed assets and their policies ($\mathsf{Asset}$, $\Phi$), and the cryptographic mechanisms that check them ($\mathsf{CS}$, $\Pi$, $\Aux$). Three design decisions shape this split:

\noindent\emph{(i)}
Policies are written over the objects of the selected \ac{FL} type rather than built into it. The same type can therefore carry different policies, and one policy can serve several types. This gives the type-aware configuration of Requirement~\textbf{(R1)}.

\noindent\emph{(ii)}
A policy is a relation over public and private values. Commitments and proofs only realise its check, so a compatible mechanism can replace them without changing what the policy means. This gives the agility of Requirement~\textbf{(R5)}.

\noindent\emph{(iii)}
Each party proves that its action satisfies the applicable policy when it sends the action. The cost of a proof therefore follows the policy. Commitments fix the private data and state that a proof refers to, so later proofs and audits refer to the same objects. Zero-knowledge proofs keep these objects hidden, and their non-interactive form lets an auditor repeat every check. This gives Requirements~\textbf{(R2)},~\textbf{(R3)} and \textbf{(R4)}.
As a result, \proto\ is neutral about the content of a policy. 
These principles determine the six framework components defined next.
\subsection{Framework Components}
\proto\ consists of six components covering the FL setting, governed objects and policies, and cryptographic mechanisms. Formally, a framework instance is represented by the tuple
\begin{equation}
\Framework =
    \left\langle
        \mathsf{Type},
        \mathsf{Asset},
        \mathsf{CS},
        \Pi,
        \Aux,
        \Phi
    \right\rangle.
  \label{eq:framework-tuple}
\end{equation}

We define the six components below:

\subsubsection{FL Type $(\mathsf{Type})$}
The federated-learning type characterises the learning setting and determines the high-level structure of the training protocol. Different FL types may differ in their participants, local computation, maintained state, and aggregation behaviour.
\begin{definition}[Federated learning type]
\label{def:type}
The component $\mathsf{Type}$ identifies the \ac{FL} setting. It is selected from a family such as
\[
\begin{aligned}
\mathcal{T}_{\mathrm{FL}}=\{&\mathsf{VanillaFL},\mathsf{PersonalizedFL},\mathsf{ContinualFL},\ldots\}.
\end{aligned}
\]
The selected type fixes the participating roles, the round structure, the local state maintained by a client, and the aggregation semantics. It does not fix a particular cryptographic construction.
\end{definition}

\subsubsection{Assets $(\mathsf{Asset})$}
A protocol operates on a collection of learning and cryptographic objects that are created, exchanged, or verified during execution. We refer to these objects collectively as the assets of the framework.
\begin{definition}[Assets]
\label{def:asset}
The set $\mathsf{Asset}$ contains the objects on which the framework operates:
\[
  \mathsf{Asset}=\{\FLObj,\CryptoObj\}.
\]
Here, $\FLObj$ contains \ac{FL} objects, such as a global model, a local model, a local update, a round identifier, and any state required by $\mathsf{Type}$. The set $\CryptoObj$ contains cryptographic objects, such as public parameters, commitments, proofs, keys, signatures, hashes, or ciphertexts.
\end{definition}

\subsubsection{Commitment Schemes $(\mathsf{CS})$}
A commitment scheme allows a party to commit to a value while keeping the value hidden, while preventing the committed value from being changed later~\cite{kzg2010commitment,pedersen1991commitment}. We use commitments to bind private objects to the values subsequently certified by the protocol.
\begin{definition}[Commitment scheme]
\label{def:cs}
The commitment scheme is 
\(
\mathsf{CS}=(\mathsf{Setup}_{\mathsf{CS}},\Commit).
\)
It provides the following algorithms:
\begin{align*}
  \mathsf{pp}_{\mathsf{CS}} \leftarrow \mathsf{Setup}_{\mathsf{CS}}(1^\lambda)
  \quad\text{and}\quad
  (c,r) \leftarrow \Commit(\mathsf{pp}_{\mathsf{CS}},m).
\end{align*}
The value $c$ binds a client to the message $m$, while $r$ is the opening
information. The scheme is required to satisfy correctness, hiding, and
binding for the security parameter $\lambda$.
When a policy requires commitments to an ordered collection
$M=(m^{(1)},\ldots,m^{(J)})$, we write
\[
  (C,R)\leftarrow\Commit(\mathsf{pp}_{\mathsf{CS}},M)
\]
as shorthand for applying the same commitment algorithm independently to
each entry, where
\[
  C=(c^{(1)},\ldots,c^{(J)})
  \quad\text{and}\quad
  R=(r^{(1)},\ldots,r^{(J)}).
\]
\end{definition}
The $\Commit$ interface in our framework intentionally returns both the commitment and its opening information to support policy-dependent use across different realisations. The opening information can remain in a private proof witness or be explicitly provided to an aggregator with the committed value when the selected policy requires opening. 
\subsubsection{\ac{NIZK} proof system $(\Pi)$}
A non-interactive zero-knowledge proof allows a prover to convince a verifier that a statement is true without revealing the corresponding private witness and without requiring interaction beyond the proof itself~\cite{blum1988nizk}. 
\begin{definition}[Non-interactive zero-knowledge proof system]
\label{def:pi}
The proof system is a NIZK system
\[
\Pi=(\mathsf{Setup}_{\Pi},\Prove,\Verify)
\]
for the policy-dependent relation family
\(
\mathcal{R}_{\Phi}
=
\{\mathcal{R}_{\phi}\}_{\phi\in\Phi}.
\)
For each $\phi\in\Phi$, its algorithms are
\[
\ppnizk \leftarrow \mathsf{Setup}_{\Pi}(1^\lambda),
\pi \leftarrow
\Pi.\Prove_{\mathcal{R}_{\phi}}
(\ppnizk,x,w)\,\text{and} 
\]
\[
b \leftarrow
\Pi.\Verify_{\mathcal{R}_{\phi}}
(\ppnizk,x,\pi).
\]

The public statement $x$ contains the public assets associated with a protocol action, while the witness $w$ contains the corresponding private assets and commitment openings. For every $\phi\in\Phi$, the NIZK proves membership in the policy-dependent relation $\mathcal{R}_{\phi}$. In particular,
\[
(x,w)\in\mathcal{R}_{\phi}
\quad\Longleftrightarrow\quad
\mathsf{Pred}_{\phi}(x,w)=1.
\]

The proof system is required to provide completeness, knowledge soundness, and zero knowledge. We state here for completeness that knowledge soundness implies soundness.
\end{definition}

\subsubsection{Auxiliary Objects ($\Aux$)}
Some realisations require additional cryptographic primitives beyond the commitment scheme and NIZK proof system. We collect such optional realisation-specific functionality under the auxiliary component.

\begin{definition}[Auxiliary objects]\label{def:aux}
The component $\Aux$ contains optional cryptographic primitives used by a realisation of the framework. Typical examples are
\[
\Aux
=
(\mathsf{Sig},\mathsf{H},\mathsf{HE},\ldots).
\]
Here, $\mathsf{Sig}$, $\mathsf{H}$, and $\mathsf{HE}$ denote optional signature, hash, and homomorphic encryption mechanisms. A realisation may use any subset of these components as required by its selected policies. Each selected auxiliary primitive must specify the correctness and security properties required by the realisation. Any additional security guarantee provided through such a primitive is conditional on its corresponding primitive-specific assumptions, which must be stated by the concrete instantiation.
\end{definition}

\subsubsection{Policy Set $(\Phi)$}
Policies specify the conditions that protocol actions are required to satisfy.
They provide a declarative way to express when a client contribution or an aggregation result is considered compliant.
\begin{definition}[Policy]\label{def:policy}
Let $\mathsf{Asset}$ denote the space of federated-learning and cryptographic assets, including the public parameters used by the protocol, and let $\Aux$ denote the space of optional realisation-specific auxiliary objects.
A policy family is a set $\Phi = \{\phi_1,\ldots,\phi_N\}$.
Each policy $\phi \in \Phi$ specifies a deterministic polynomial-time predicate
\[
    \mathsf{Pred}_{\phi} :
    \mathcal{X}_{\phi} \times \mathcal{W}_{\phi}
    \longrightarrow \{0,1\},
\]
where
$\mathcal{X}_{\phi}\subseteq
\mathsf{Asset}^{*}\times\mathsf{Aux}^{*}$
is the space of public statements, and
$\mathcal{W}_{\phi}\subseteq
\mathsf{Asset}^{*}\times\mathsf{Aux}^{*}$
is the space of private witnesses.

For $x\in\mathcal{X}_{\phi}$ and
$w\in\mathcal{W}_{\phi}$, the value
\(
    \mathsf{Pred}_{\phi}(x, w)=1
\)
means that the assets and auxiliary objects represented by $(x,w)$ satisfy policy $\phi$. A policy may govern either a client action or an aggregation action. Its predicate includes any required consistency among the public and private assets represented by the statement and witness.
\end{definition}

\subsection{Framework Realisation}
The components defined above specify the building blocks of a framework instance. A concrete realisation combines compatible choices for these components into an executable policy-compliant FL protocol and specifies how participants construct, certify, and verify policy-compliant actions. We first translate each policy predicate into the relation recognised by the proof system. We then describe the generic setup and the client- and aggregator-side procedures that use these relations during protocol execution.
\subsubsection{Policy Relation Transformation}
The proof system certifies membership in a relation. We therefore associate each policy predicate with a relation and its corresponding language.
\begin{definition}[Policy relation]
For a policy $\phi \in \Phi$, the policy predicate
$\mathsf{Pred}_{\phi}$ induces the relation
\[
    R_{\phi}
    =
    \left\{
        (x,w)
        \in
        \mathcal{X}_{\phi} \times \mathcal{W}_{\phi}
        \;\middle|\;
        \mathsf{Pred}_{\phi}
        (x,w) = 1
    \right\}.
\]

The language associated with $R_{\phi}$ is
\[
    L_{\phi}
    =
    \left\{
        x \in \mathcal{X}_{\phi}
        \;\middle|\;
        \exists\,w \in \mathcal{W}_{\phi}
        \text{ such that }
        (x,w) \in R_{\phi}
    \right\}.
\]

Therefore, proving compliance with policy $\phi$ amounts to proving the statement
\(
    x \in L_{\phi},
\)
or equivalently,
\[
    \exists\,w \in \mathcal{W}_{\phi} :
    \mathsf{Pred}_{\phi}(x,w) = 1.
\]
The same transformation is used for client and aggregation policies. In the generic algorithms below, $\phi\in\Phi$ denotes the client policy and $\psi\in\Phi$ denotes the aggregation policy. Both belong to the same policy family $\Phi$.
\end{definition}

\subsection{Framework Generic Algorithms}\label{sec:generic-algo}
The policy relations defined above are used by the framework through three generic procedures: (i) a one-time setup procedure, (ii) a client-side procedure executed in each training round, and (iii) an aggregator-side procedure. 

The named procedures appearing in the algorithms, such as $\LocalTrain$, $\Aggregate$, $\Update$, and the auxiliary operations, are abstract interfaces. Their concrete behaviour is determined by the selected
$\mathsf{Type}$, policy family, and realisation. Thus, the algorithms specify how learning and cryptographic components are composed without fixing a particular FL algorithm or cryptographic implementation. The symbol $\bot$ denotes failure, or the absence of an optional output when indicated by the context.

\subsubsection{Generic framework Setup}
Algorithm~\ref{alg:generic-framework-setup-compact} initialises the framework before training begins. Let
\(
\mathcal P
=
\{\mathsf{C}^{(i)}\}_{i\in I}
\cup
\{A^{(j)}\}_{j\in J}
\)
denote the set of participating clients and aggregators. The setup derives the public parameters required by the selected realisation and initialises any participant-specific state required by its cryptographic mechanisms.

The realisation specification $\mathsf{spec}$ contains choices that are not fixed directly by the framework tuple, such as model dimensions, identifiers, thresholds, and realisation-specific procedure parameters. The procedure $\operatorname{ParamGen}$ derives the corresponding common parameters $\mathrm{params}$. The commitment scheme and NIZK proof system are then initialised independently to obtain
$\mathrm{pp}_{\mathsf{CS}}$ and $\mathrm{pp}_{\Pi}$. 

A realisation may additionally use optional mechanisms from
$\Aux$. We write
\(
\Aux=(\Aux_1,\ldots,\Aux_\ell),
\)
where $\Aux_k$ denotes the $k$-th optional mechanism used by the realisation. For each such mechanism,
$\operatorname{Setup}_{\mathsf{Aux}_k}$ produces its public parameters $\mathrm{pp}_{\mathsf{Aux}_k}$ and, when required, private state $\mathrm{st}^{\mathsf{Aux}_k}_P$ for each participant $P\in\mathcal P$.
A component that requires no participant-specific state returns $\bot$ for that state, while a component requiring no setup may return a fixed public description. 

The framework combines these outputs into the public parameter set $\mathrm{pp}_{\mathcal F}$ and the participant-specific setup state $\mathrm{st}^{\mathrm{setup}}_P$. Public parameters are made available to
all participants, whereas private setup state is delivered only to the participant for which it was generated. We assume authenticated and confidential channels for this setup distribution and their concrete establishment is deployment-specific and lies outside the generic framework.

After setup, $\Params$ denotes the public portion of $\mathrm{pp}_{\mathcal F}$ required by a particular round operation. Temporary setup randomness and any trapdoor information not explicitly
assigned to a participant are erased. In the round algorithms below, a participant is understood to have access to its assigned setup state $\mathrm{st}^{\mathrm{setup}}_P$ whenever required by a selected auxiliary mechanism and this state is suppressed from the algorithm inputs for readability.

\begin{algorithm}[t]
\caption{Generic Setup for \proto}
\label{alg:generic-framework-setup-compact}
\small
\begin{algorithmic}[1]
\Require Security param $1^\lambda$, 
$\mathcal F=(\mathsf{Type},\mathsf{Asset},\mathsf{CS},
\Pi,\mathsf{Aux},\Phi)$, participant sets
$\{\mathsf{C}^{(i)}\}_{i\in I}$ and $\{A^{(j)}\}_{j\in J}$,
realisation specification $\mathsf{spec}$
\Ensure $\mathrm{pp}_{\mathcal F}$ and
$\{\mathrm{st}^{\mathrm{setup}}_P\}_{P\in\mathcal P}$

\State $\mathcal P\gets
       \{\mathsf{C}^{(i)}\}_{i\in I}\cup
       \{A^{(j)}\}_{j\in J}$
\State Establish the required secure channels among $\mathcal P$

\State $\mathrm{params}\gets
       \operatorname{ParamGen}(1^\lambda,\mathsf{Type},
       \mathsf{Asset},\Phi,\mathsf{spec})$
\State $\mathrm{pp}_{\mathsf{CS}}\gets
       \operatorname{Setup}_{\mathsf{CS}}(1^\lambda)$
\State $\mathrm{pp}_{\Pi}\gets
       \operatorname{Setup}_{\Pi}(1^\lambda)$

\For{$k=1,\ldots,\ell$}
    \State
    $\left(
      \mathrm{pp}_{\mathsf{Aux}_k},
      \{\mathrm{st}^{\mathsf{Aux}_k}_P\}_{P\in\mathcal P}
    \right)
    \gets
    \operatorname{Setup}_{\mathsf{Aux}_k}
    (1^\lambda,\mathcal P,\mathrm{params})$
    \If{$\operatorname{Setup}_{\mathsf{Aux}_k}$ fails}
        \State \Return $\bot$
    \EndIf
\EndFor

\State $\mathrm{pp}_{\mathsf{Aux}}\gets
       (\mathrm{pp}_{\mathsf{Aux}_1},\ldots,
        \mathrm{pp}_{\mathsf{Aux}_{\ell}})$

\ForAll{$P\in\mathcal P$}
    \State $\mathrm{st}^{\mathrm{setup}}_P\gets
           (\mathrm{st}^{\mathsf{Aux}_1}_P,\ldots,
            \mathrm{st}^{\mathsf{Aux}_{\ell}}_P)$
\EndFor

\State $\mathrm{pp}_{\mathcal F}\gets
       (\mathrm{params},
        \mathrm{pp}_{\mathsf{CS}},
        \mathrm{pp}_{\Pi},
        \mathrm{pp}_{\mathsf{Aux}})$

\State Publish $\mathrm{pp}_{\mathcal F}$ and securely provide
       $\mathrm{st}^{\mathrm{setup}}_P$ to each $P\in\mathcal P$
\State Erase temporary setup randomness and unprescribed trapdoors

\State \Return
       $\left(
        \mathrm{pp}_{\mathcal F},
        \{\mathrm{st}^{\mathrm{setup}}_P\}_{P\in\mathcal P}
       \right)$
\end{algorithmic}
\end{algorithm}

\subsubsection{Client Side execution}
Algorithm~\ref{alg:client} describes how client $\client{i}$ constructs a proof-carrying message for round $t$. The client first executes the local learning procedure prescribed by $\mathsf{Type}$ to obtain its local contribution
\[
u_{i,t}
\gets
\LocalTrain(\mathsf{Type},\theta_t,D_{i,t},\Params).
\]
The object $\FLObj_{i,t}$ collects the learning state relevant to the client policy, including the current model, local data, resulting contribution, client identity, and round identifier.

For a client policy $\phi$, let
\(
M_{i,t}
=
(m_{i,t}^{(1)},\ldots,m_{i,t}^{(J_{\phi})})
\)
denote the ordered collection of objects that the policy requires the client to commit to. Both the number of entries $J_{\phi}$ and their meaning are determined by $\phi$. The commitment procedure produces the public commitments $C_{i,t}$ and the corresponding private opening information
$R_{i,t}$.

The policy may also require auxiliary mechanisms. We write $\Aux_{\phi}$ for the subset of auxiliary mechanisms used by $\phi$. Applying these mechanisms to $\FLObj_{i,t}$ produces public outputs
$Z_{i,t}^{\mathsf{pub}}$ and private outputs $Z_{i,t}^{\mathsf{priv}}$ and either collection may be empty when no such output is required.

The client then forms the public statement
\[
x_{i,t}
=
(i,t,\phi,\theta_t,C_{i,t},
 Z_{i,t}^{\mathsf{pub}},\Params)
\]
and the private witness
\(
w_{i,t}
=
(\FLObj_{i,t},M_{i,t},R_{i,t},
 Z_{i,t}^{\mathsf{priv}}).
\)
This separation ensures that the verifier receives only the public context, commitments, and permitted auxiliary outputs, while the underlying learning objects, commitment openings, and private auxiliary state remain in the witness.

Before producing a proof, the client checks that $(x_{i,t},w_{i,t})\in\mathcal{R}_{\phi}$. If the relation is not satisfied, the client aborts and returns $\bot$. Otherwise, it generates a NIZK proof
$\pi_{i,t}$ for $\mathcal{R}_{\phi}$ and sends
\(
\mu_{i,t}
=
(i,t,x_{i,t},\pi_{i,t}).
\)
The witness $w_{i,t}$ is never transmitted.
\begin{algorithm}[t]
\caption{Client $\client{i}$ in round $t$}
\label{alg:client}
\small
\begin{algorithmic}[1]
\Require $\Framework$, $\phi\in\Phi$, $\theta_t$, and $D_{i,t}$
\Ensure Client message $\mu_{i,t}$ or $\bot$

\State $u_{i,t}\gets
       \LocalTrain(\mathsf{Type},\theta_t,D_{i,t},\Params)$
\State Construct $\FLObj_{i,t}$ from
       $(\theta_t,D_{i,t},u_{i,t},i,t)$
\State Construct $M_{i,t}$ from $\FLObj_{i,t}$ as prescribed by $\phi$
\State $(C_{i,t},R_{i,t})\gets
       \Commit(\mathsf{pp}_{\mathsf{CS}},M_{i,t})$
\State $(Z_{i,t}^{\mathsf{pub}},Z_{i,t}^{\mathsf{priv}})\gets
       \mathsf{Apply}(\Aux_{\phi},\FLObj_{i,t},\Params)$

\State $x_{i,t}\gets
       (i,t,\phi,\theta_t,C_{i,t},
        Z_{i,t}^{\mathsf{pub}},\Params)$
\State $w_{i,t}\gets
       (\FLObj_{i,t},M_{i,t},R_{i,t},
        Z_{i,t}^{\mathsf{priv}})$

\If{$\mathcal{R}_{\phi}(x_{i,t},w_{i,t})=0$}
    \State \Return $\bot$
\EndIf

\State $\pi_{i,t}\gets
       \Pi.\Prove_{\mathcal{R}_{\phi}}(\ppnizk,x_{i,t},w_{i,t})$
\State \Return
       $\mu_{i,t}=(i,t,x_{i,t},\pi_{i,t})$
\end{algorithmic}
\end{algorithm}

\subsubsection{Aggregator side execution}
Algorithm~\ref{alg:aggregator} describes the execution of aggregator $\agg{\ell}$ in round $t$. It uses the client policy $\phi$ to determine which submitted client messages may be admitted and the aggregation policy $\psi$ to determine whether the resulting aggregation action is compliant.
The aggregator maintains three synchronised records. The set $I_t^{\mathsf{adm}}$ contains the identifiers of admitted clients, $\mathsf{Adm}_t$ is the ordered record of their accepted proof-carrying
messages, and $V_t$ contains the corresponding values used for numerical aggregation.

For each received message
\(
\mu_{i,t}=(i,t,x_{i,t},\pi_{i,t}),
\)
the aggregator first rejects duplicate client identifiers. It then checks that the public statement is bound to the expected client, round, policy, current model, and public parameters through $\mathsf{Ctx}$. The accompanying NIZK proof is verified under $\Pi$. Only after these checks succeed does the 
aggregator call
\(
\mathsf{ExtractAgg}_{\mathsf{Type},\phi}(x_{i,t})
\)
to obtain the contribution $v_{i,t}$ prescribed by the selected FL type and client policy. If no valid aggregation value can be obtained, the message is discarded. Otherwise, the client identifier, message, and extracted value are added consistently to $I_t^{\mathsf{adm}}$, $\mathsf{Adm}_t$, and $V_t$.

Once all candidate messages have been processed, the participation predicate $\mathsf{Part}_{\mathsf{Type},\psi}$ checks whether the admitted set satisfies the participation requirements of the selected FL type and aggregation policy. If this condition fails, the round aborts.
The aggregator then computes the round aggregate. The procedure $\Aggregate_{\mathsf{Type},\psi}$ specifies the aggregation rule prescribed by the FL type and policy, while $\mathsf{ApplyAgg}$ allows the realisation to carry out that rule using any auxiliary mechanism required by $\psi$.
When no auxiliary mechanism is required, it reduces to the prescribed aggregation computation itself. The resulting aggregate $\Delta_t$ is applied to the current model by $\Update_{\mathsf{Type}}$ to obtain $\theta_{t+1}$.

The aggregation policy may require additional auxiliary evidence or state.
The subsequent call to $\mathsf{Apply}$ therefore produces public outputs $Z_t^{\mathsf{pub}}$ and private outputs $Z_t^{\mathsf{priv}}$ used in the aggregation statement and witness. The aggregator forms
\(
x_t^{\mathsf{agg}}
=
(\ell,t,\phi,\psi,\theta_t,\mathsf{Adm}_t,
 \theta_{t+1},Z_t^{\mathsf{pub}},\Params)
\)
and
\(
w_t^{\mathsf{agg}}
=
(V_t,\Delta_t,Z_t^{\mathsf{priv}}).
\)

The relation $\mathcal{R}_{\psi}$ binds the aggregation result to the admitted record and to the computation prescribed by the aggregation policy. If $(x_t^{\mathsf{agg}},w_t^{\mathsf{agg}})\notin\mathcal{R}_{\psi}$, the aggregator aborts. Otherwise, it generates $\pi_t^{\mathsf{agg}}$ and returns the proof-carrying aggregation message
\(
\nu_t
=
(\ell,t,x_t^{\mathsf{agg}},\pi_t^{\mathsf{agg}}).
\)
A verifier of $\nu_t$ first verifies the client proofs contained in $\mathsf{Adm}_t$ and then verifies the aggregation proof $\pi_t^{\mathsf{agg}}$. Hence, the aggregation proof certifies the aggregation
action relative to the publicly recorded admitted messages and it does not need to recursively verify the client proofs within the aggregation circuit.

\begin{algorithm}[t]
\caption{Aggregator $\agg{\ell}$ in round $t$}
\label{alg:aggregator}
\small
\begin{algorithmic}[1]
\Require $\Framework$, $\phi,\psi\in\Phi$, $\theta_t$, and
         $\{\mu_{i,t}\}_{i\in I_t}$
\Ensure Aggregation message $\nu_t$ or $\bot$

\State $I_t^{\mathsf{adm}}\gets\emptyset$,
       $\mathsf{Adm}_t\gets()$, and $V_t\gets()$

\ForAll{$\mu_{i,t}=(i,t,x_{i,t},\pi_{i,t})$}
    \If{$i\in I_t^{\mathsf{adm}}
         \lor \mathsf{Ctx}(x_{i,t})
         \neq(i,t,\phi,\theta_t,\Params)$}
        \State \textit{skip to next client message}
    \EndIf
    \If{$\Pi.\Verify_{\mathcal{R}_{\phi}}
          (\ppnizk,x_{i,t},\pi_{i,t})=0$}
        \State \textit{skip to next client message}
    \EndIf

    \State $v_{i,t}\gets
           \mathsf{ExtractAgg}_{\mathsf{Type},\phi}(x_{i,t})$
    \If{$v_{i,t}=\bot$}
        \State \textit{skip to next client message}
    \EndIf

    \State $I_t^{\mathsf{adm}}\gets
           I_t^{\mathsf{adm}}\cup\{i\}$ and $V_t[i]\gets v_{i,t}$
    \State Append $\mu_{i,t}$ to $\mathsf{Adm}_t$
\EndFor

\If{$\mathsf{Part}_{\mathsf{Type},\psi}
      (I_t^{\mathsf{adm}},\Params)=0$}
    \State \Return $\bot$
\EndIf

\State $\Delta_t\gets
       \mathsf{ApplyAgg}
       (\Aux_{\psi},\Aggregate_{\mathsf{Type},\psi},
        V_t,\Params)$
\State $\theta_{t+1}\gets
       \Update_{\mathsf{Type}}
       (\theta_t,\Delta_t,\Params)$
\State $(Z_t^{\mathsf{pub}},Z_t^{\mathsf{priv}})\gets
       \mathsf{Apply}
       (\Aux_{\psi},\theta_t,\mathsf{Adm}_t,V_t,
        \Delta_t,\theta_{t+1},\Params)$

\State $x_t^{\mathsf{agg}}\gets
       (\ell,t,\phi,\psi,\theta_t,\mathsf{Adm}_t,
        \theta_{t+1},Z_t^{\mathsf{pub}},\Params)$
\State $w_t^{\mathsf{agg}}\gets
       (V_t,\Delta_t,Z_t^{\mathsf{priv}})$

\If{$\mathcal{R}_{\psi}
      (x_t^{\mathsf{agg}},w_t^{\mathsf{agg}})=0$}
    \State \Return $\bot$
\EndIf

\State $\pi_t^{\mathsf{agg}}\gets
       \Pi.\Prove_{\mathcal{R}_{\psi}}(\ppnizk,x_t^{\mathsf{agg}},w_t^{\mathsf{agg}})$
\State \Return
       $\nu_t=(\ell,t,x_t^{\mathsf{agg}},\pi_t^{\mathsf{agg}})$
\end{algorithmic}
\end{algorithm}

\section{Security Analysis}\label{sec:security-analysis}
We establish update admissibility, aggregation correctness, and policy-relative privacy for one protocol round. 
\textcolor{black}{For polynomially many clients and rounds, update admissibility and aggregation correctness follow by a union bound over all accepted messages. 
Policy-relative privacy follows by a hybrid argument that uses the multi-theorem zero knowledge of Appendix~\ref{app:crypto-notions}.}
Sender identities are authenticated as specified in the threat model, and a corrupted participant receives its prescribed setup state.

\noindent\textbf{Cryptographic assumptions.}
Our generic guarantees rely on standard properties of the commitment scheme and NIZK proof system. Soundness of $\Pi$ is used to establish update admissibility and aggregation correctness, while zero knowledge of $\Pi$ together with hiding of $\mathsf{CS}$ establishes policy-relative privacy.
\textcolor{black}{For client policies over committed private data or retained state, knowledge soundness of $\Pi$ and binding of $\mathsf{CS}$ ensure that the policy holds on the committed objects.}
Completeness and correctness ensure acceptance of honestly generated messages.
These are the assumptions used by the generic security reductions below.
When a concrete realisation provides an additional guarantee through an auxiliary cryptographic primitive, that guarantee is additionally subject to the primitive-specific assumptions stated by that realisation.
Formal definitions of the cryptographic notions are given in Appendix~\ref{app:crypto-notions}. Shared notation used by the framework security games is given in Appendix~\ref{app:security-game-notions}.

\noindent\textbf{Update admissibility.}
A malicious client succeeds if it produces a message that passes the framework's client-verification procedure while being semantically invalid under the active client policy. The formal experiment is given in Appendix~\ref{app:game-adm}.

\begin{theorem}[Update admissibility]
\label{thm:update-adm}
For every $\mathsf{PPT}$ adversary $\mathcal{A}_1$, there exists a $\mathsf{PPT}$ adversary $\mathcal{B}_{\mathsf{adm}}$ such that
\[
\operatorname{Adv}^{\mathsf{adm}}_{\Framework,\mathcal{A}_1}(\lambda)
\leq
\operatorname{Adv}^{\mathsf{snd}}_{\Pi,\mathcal{B}_{\mathsf{adm}}}(\lambda).
\]
Consequently, the framework provides update admissibility whenever $\Pi$ is sound.
\end{theorem}

The formal experiment, advantage definition, and proof are given in Appendix~\ref{app:game-adm}.

\noindent\textbf{Aggregation correctness.}
A malicious aggregator succeeds if it produces an aggregation message that passes framework verification but is not semantically valid with respect to the active aggregation policy and recorded admitted set. The formal experiment is given in Appendix~\ref{app:game-agg}.

\begin{theorem}[Aggregation correctness]
\label{thm:agg-correct}
Let at most $q(\lambda)$ client messages appear in $\mathsf{Rec}_t$ ($\mathsf{Adm}_t$ in Algorithm~\ref{alg:aggregator}).
For every $\mathsf{PPT}$ adversary $\mathcal{A}_2$, there exist $\mathsf{PPT}$ adversaries
$\mathcal{B}_{1},\ldots,\mathcal{B}_{q}$ and $\mathcal{B}_{\mathsf{agg}}$ such that
\[
\operatorname{Adv}^{\mathsf{agg}}_{\Framework,\mathcal{A}_2}(\lambda)
\leq
\sum_{j=1}^{q(\lambda)}
\operatorname{Adv}^{\mathsf{snd}}_{\Pi,\mathcal{B}_{j}}(\lambda)
+
\operatorname{Adv}^{\mathsf{snd}}_{\Pi,\mathcal{B}_{\mathsf{agg}}}(\lambda).
\]
Hence the aggregation-correctness advantage is negligible when $q$ is polynomial and $\Pi$ is sound.
\end{theorem}

The formal experiment, advantage definition, and proof are given in Appendix~\ref{app:game-agg}.

\noindent\textbf{Policy-relative privacy.}
The privacy experiment compares two valid private client inputs that induce the same permitted public context and asks whether a verifier can distinguish which input generated the proof-carrying message. The formal experiment is given in Appendix~\ref{app:game-priv}.

\begin{theorem}[Policy-relative privacy]
\label{thm:policy-privacy}
For every $\mathsf{PPT}$ verifier $\mathcal{A}_2$, there exist
$\mathsf{PPT}$ adversaries
$\mathcal{B}_{\mathsf{zk},0}$,
$\mathcal{B}_{\mathsf{hid}}$, and
$\mathcal{B}_{\mathsf{zk},1}$ such that
\[
\begin{aligned}
\operatorname{Adv}^{\mathsf{priv}}_{\Framework,\mathcal{A}_2}(\lambda)
\leq{}&
\operatorname{Adv}^{\mathsf{zk}}_{\Pi,\mathcal{B}_{\mathsf{zk},0}}(\lambda)
+
\operatorname{Adv}^{\mathsf{hid}}_{\mathsf{CS},\mathcal{B}_{\mathsf{hid}}}
(\lambda)
\\
&+
\operatorname{Adv}^{\mathsf{zk}}_{\Pi,\mathcal{B}_{\mathsf{zk},1}}(\lambda).
\end{aligned}
\]
Consequently, the framework provides policy-relative privacy when $\Pi$ is zero knowledge and $\mathsf{CS}$ is hiding.
\end{theorem}

The formal experiment, advantage definition, and proof are given in Appendix~\ref{app:game-priv}.

\section{Instantiation}\label{sec:Instantiation}
We now return to the hospital consortium of Sec.~\ref{sec:introduction}. Each instantiation answers one of its four compliance questions and shows how the abstract interfaces of Sec.~\ref{sec:framework} are specialised for the corresponding \ac{FL} setting. VanillaFL answers (Q1) and is presented here. ContinualFL, PersonalizedFL, and a Threshold-HE realisation answer (Q2) to (Q4) and are given in Appendix~\ref{app:additional-instantiations}. Table~\ref{tab:regulatory-mapping} links each question to the regulatory expectation behind it. The four client policies differ on purpose. They are a loss condition on private data, a retention bound relative to retained state, an optimality condition on a private intermediate model, and a norm bound on an encrypted update. For each realisation, we specify the type-specific learning and state semantics, the client policy, the value extracted for aggregation, the aggregation policy, and any required auxiliary mechanism. Generic proof and message construction follow Algorithms~\ref{alg:client} and~\ref{alg:aggregator}.

\noindent{\bf Vanilla Federated Learning}
\begin{question}{(Q1)}
    Does each contribution reach the agreed quality level on the hospital's own data, and does it count with the weight that this data entitles it to? The party that aggregates the model cannot inspect the training data~\cite{edps2025fl}. A hospital could therefore submit a contribution that is unrelated to its data, or claim a larger dataset to obtain a larger weight. The policy below turns the agreed quality criterion into a condition that every admitted contribution provably meets. It supports the data-governance expectation in Table~\ref{tab:regulatory-mapping}.
\end{question}
We instantiate $\mathsf{Type}=\mathsf{VanillaFL}$ with clients $\{\client{i}\}_{i\in I}$ and one aggregator $\agg{1}$~\cite{yangFL2019}. In round $t$, the aggregator distributes the current global model $\theta_t$. Each client $\client{i}$ holds a local dataset $D_i$ of size $n_i=|D_i|$ and computes
\(
u_{i,t}
\gets
\LocalTrain(\mathsf{VanillaFL},\theta_t,D_i,\Params),
\) and
\(
\theta_{i,t}
=
\theta_t+u_{i,t}.
\)

\emph{Client policy.}
The client policy $\phi_{\mathsf{loss}}\in\Phi$ requires the submitted update to be bound to the client's private dataset and the resulting local model to satisfy a public loss threshold. The committed collection is
\(
M_{i,t}
=
(D_i,u_{i,t})
\) and
\(
(C_{i,t},R_{i,t})
\gets
\Commit(\mathsf{pp}_{\mathsf{CS}},M_{i,t}).
\)
The update $u_{i,t}$ and dataset size $n_i$ are public, whereas $D_i$ and $R_{i,t}$ remain private. No auxiliary mechanism is required.

The client predicate is
\[
\boxed{
\begin{aligned}
&\mathsf{Pred}_{\phi_{\mathsf{loss}}}(x_{i,t},w_{i,t})=1
\Longleftrightarrow
C_{i,t}\text{ binds }(D_i,u_{i,t})
\\
&\quad{}\land n_i=|D_i|
\land \mathsf{Loss}(\theta_t+u_{i,t},D_i)\leq\tau .
\end{aligned}
}
\]
Thus, the relation binds the submitted update to the committed dataset and checks the loss condition on the resulting local model. It also fixes the public weight $n_i$ to the size of the committed dataset, and the verifier learns nothing else about $D_i$. Execution of $\LocalTrain$ itself is outside this policy.
Following Algorithm~\ref{alg:client}, the client proves the corresponding
relation and returns the proof-carrying message
\[
\pi_{i,t}
\gets
\Pi.\Prove_{\mathcal{R}_{\phi_{\mathsf{loss}}}}
(\ppnizk,x_{i,t},w_{i,t}),
\,
\mu_{i,t}
=
(i,t,x_{i,t},\pi_{i,t}).
\]

\emph{Aggregation policy.}
For every admitted client, the type-and-policy extraction interface returns
\(
v_{i,t}
=
\mathsf{ExtractAgg}_{\mathsf{VanillaFL},\phi_{\mathsf{loss}}}
(x_{i,t})
=
(u_{i,t},n_i),
\)
and
\(
V_t[i]
=
(u_{i,t},n_i).
\)
The aggregation policy
$\psi_{\mathsf{FedAvg}}\in\Phi$ applies weighted FedAvg~\cite{fedavgmcmahan17a} to these
extracted values. Let
\(
N_t
=
\sum_{i\in I_t^{\mathsf{adm}}}n_i.
\)
The prescribed aggregate and next global model are
\(
\Delta_t
=
\sum_{i\in I_t^{\mathsf{adm}}}
\frac{n_i}{N_t}u_{i,t},
\)
and 
\(
\theta_{t+1}
=
\theta_t+\Delta_t.
\)

Accordingly,
\(
\mathsf{Pred}_{\psi_{\mathsf{FedAvg}}}
(x_t^{\mathsf{agg}},w_t^{\mathsf{agg}})=1
\)
requires $V_t$ to contain exactly the values extracted from $\mathsf{Adm}_t$, $\Delta_t$ to equal the prescribed FedAvg aggregate, and $\theta_{t+1}$ to be obtained from $\theta_t$ and $\Delta_t$.
Following Algorithm~\ref{alg:aggregator}, the aggregator proves $\mathcal{R}_{\psi_{\mathsf{FedAvg}}}$ and returns 
\(
\pi_t^{\mathsf{agg}}
\gets
\Pi.\Prove_{\mathcal{R}_{\psi_{\mathsf{FedAvg}}}}
(\ppnizk,x_t^{\mathsf{agg}},w_t^{\mathsf{agg}})
\) and 
\(
\nu_t
=
(1,t,x_t^{\mathsf{agg}},\pi_t^{\mathsf{agg}}).
\)

\emph{Answer to (Q1).}
Every contribution in the aggregate met the agreed loss criterion on committed data and was assigned the weight $n_i/N_t$ that its proven dataset size entitles it to. The recorded messages of the round are the evidence, and an auditor can check them again without access to the data of any hospital.

\section{Evaluation}\label{sec:evaluation}

We evaluate four proof-of-concept realisations of the framework described in Sec.~\ref{sec:Instantiation} and Appendix~\ref{app:additional-instantiations}. Our objective is to demonstrate end-to-end policy verification and quantify its computational and communication overhead, rather than to optimise model accuracy or network performance.

\subsection{Proof-of-Concept Implementation}\label{sec:poc-implementation}

{\bf Implementation.}
We implement the VanillaFL, ContinualFL, and PersonalizedFL realisations using Flower~\cite{beutel2020flower} for \ac{FL} execution and EZKL~\cite{ezkl} for policy-proof generation and verification. The orchestration and experimental logic are implemented in Python, while EZKL uses its Rust backend to construct circuits and generate proofs. The Threshold-HE realisation instead uses its native group, range, and threshold-decryption proof mechanisms. Further implementation details are
given in Appendix~\ref{app:threshold-he-instantiation}. The Threshold-HE research prototype uses trusted-dealer key generation, and key-generation and setup costs are excluded from the reported round-level measurements.

{\bf Testbed and dataset.}
We evaluate the impact of policy enforcement on model accuracy on a local laptop with an Intel Core Ultra 7 165H processor and 32~GB RAM. Cryptographic microbenchmarks run separately on an AWS EC2 instance with an Intel Xeon Platinum 8488C processor, 32 vCPUs, and approximately 248~GiB RAM. Clients execute sequentially, generating one proof at a time with EZKL using 32 CPU threads.
We evaluate our implementation on MNIST~\cite{lecun1998gradient} with $28\times28$ grayscale images across 100 clients. For policy-aware testing, each client receives 600 distinct training images and uses a batch size of 6 fresh images per round, without cross-client overlap or reuse. VanillaFL uses an IID partition, while PersonalizedFL uses two label-sorted shards per client, with test data matching each client's training-label distribution and remaining disjoint from training data. ContinualFL uses IID partitions with four sequential Permuted-MNIST tasks of 25 rounds each, where the first retains the original pixel order and subsequent tasks use fixed permutations with unchanged labels; earlier tasks are reevaluated to measure forgetting. For EZKL benchmarks, the circuit input uses 4 images for VanillaFL and ContinualFL, and 2 images for PersonalizedFL.



{\bf Experimental settings.}
We have two experimental settings:

{\emph{(i)~Policy-Aware Setting.}} In this setting, we execute 100 clients over 100 FL rounds to present the changes in accuracy with \emph{policy compliance} and without policy enforcement. These experiments do not include cryptographic proof generation and verification,, and they are evaluated separately through the microbenchmarks.
To demonstrate this, we allow $20\%$ of the clients to behave maliciously in two separate rounds and introduce a fixed-target model-poisoning attack. Each malicious client replaces its trained end model with an input-independent classifier that always predicts digit $0$, while retaining the update structure required by the associated FL type, e.g., the recurrence state in PersonalizedFL or the Fisher-information state in ContinualFL. We fix the attack rounds across both the policy-enforced and non-policy-enforced settings to ensure a fair comparison. In the policy-enforced setting, the server aggregates only updates that are policy-compliant, whereas without policy enforcement, the server aggregates every received update. Furthermore, the policy threshold is selected by repetitive experiments. 
However, identifying an optimal threshold for every FL type is outside the scope of this paper. Instead, \proto\ demonstrates the necessity of enforcing policy compliance and its effect on preventing non-compliant updates from influencing the global learning process.

{\emph{(ii)~Microbenchmarking.}} We benchmark our \ac{PoC} with 5 clients and 5 rounds due to resource limitations and we report the mean and sample standard deviation for each measurement. Consequently, the client-proof results are based on 25 proof-generation measurements (five clients in each of 5 rounds), whereas the aggregation-proof and cohort-admission results are each based on 5 round-level measurements (one measurement per round). Throughout the evaluation, we use the term \emph{evidence} to denote proofs and supporting cryptographic data. Proof generation includes witness construction and proving, but excludes circuit setup, local training, and prover-side proof verification. Cohort admission measures verification of all five client submissions in a round. Threshold-HE uses five clients over five measured rounds, yielding 25 client measurements. Three helpers participate in 3-of-3 threshold decryption, and aggregation and independent saved-round validation each provide five measurements. All roles execute sequentially on one machine, and all measurements exclude transport and network costs. The reported timings therefore represent computational overhead rather than parallel, end-to-end round latency.

\begin{figure*}[t]
  \centering  \includegraphics[scale=0.27]{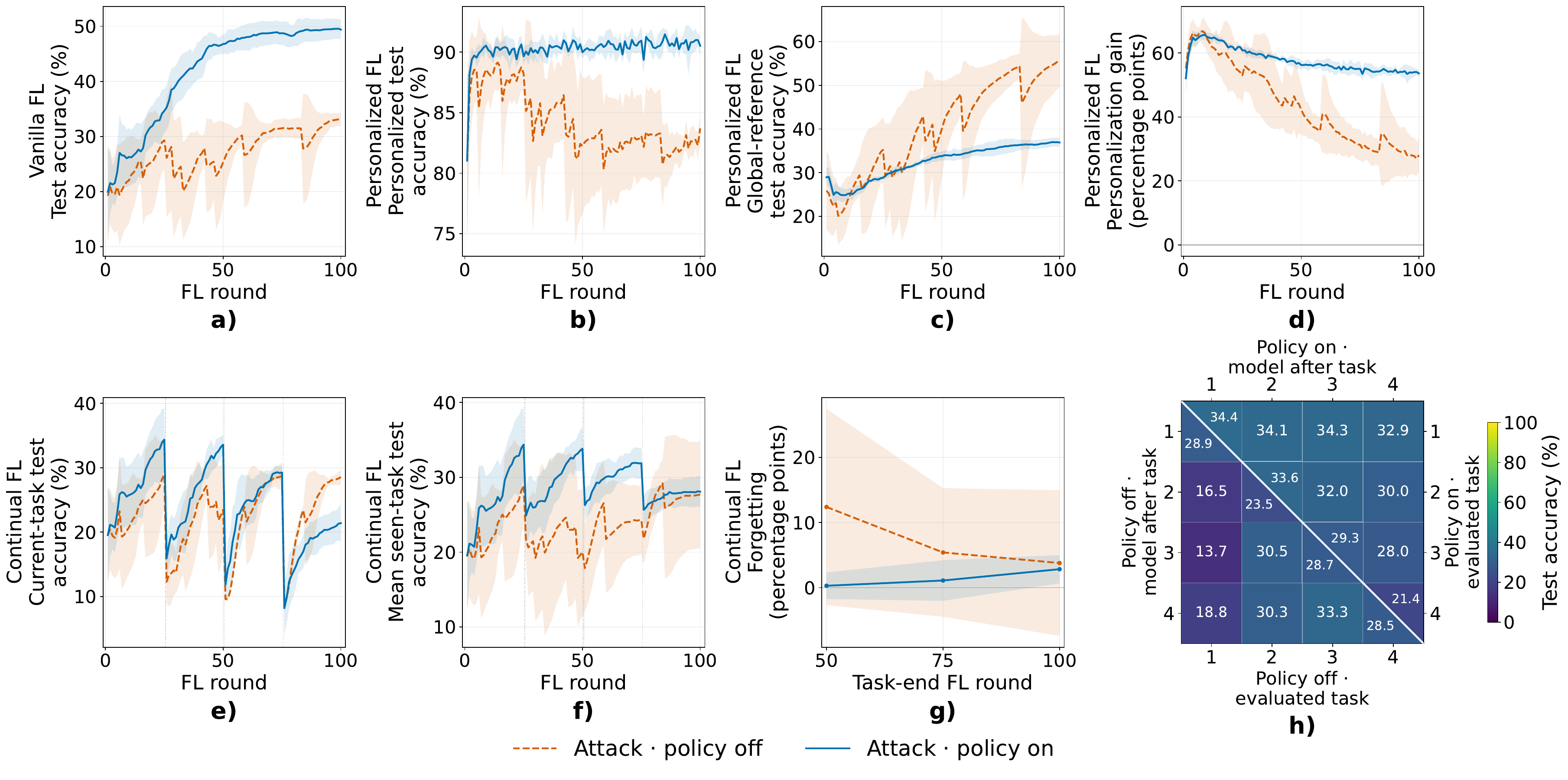}
    \caption{Effect of policy enforcement under model-poisoning attacks in vanilla (a), personalised (b-d), and continual (e-h) FL.
    Shading denotes standard deviation across five paired seeds. In (h), upper/lower triangles show policy enforcement on/off,
    respectively.}
  \label{fig:overview}
\end{figure*}

\subsection{Effect and Overhead of Policy Enforcement}\label{sec:performance-evaluation}

{\bf Effect on the learning objectives.}
The policy-aware experiments ask whether enforcing an agreed policy protects the learning objective that the policy stands for. Each objective belongs to one compliance question and one regulatory expectation in Table~\ref{tab:regulatory-mapping}. Fig.~\ref{fig:overview} compares policy enforcement (blue) with no enforcement (orange) over five paired seeds for the non-compliant contributions of Sec.~\ref{sec:poc-implementation}. The shading indicates standard deviation. In all paired comparisons, we list policy-enforced results first.

\emph{(Q1) Contribution quality.}
The objective in VanillaFL is the accuracy of the shared model. Accuracy at round~100 is $49.4\%$ with enforcement and $33.1\%$ without, representing a gain of $16.3$ percentage points (Fig.~\ref{fig:overview}(a)). Without enforcement, the model shows repeated declines and slower recovery. With enforcement, contributions that fail the loss condition do not enter the aggregate, and training is more stable.

\emph{(Q2) Controlled change.}
The objective in ContinualFL is to learn new tasks while retaining earlier ones. Clients train on each of four sequential tasks for $25$ rounds. Forgetting after Tasks~2 to 4 is approximately $(0,1,2.5)$ percentage points with enforcement, compared to $(12,5,3.5)$ without (Fig.~\ref{fig:overview}(g)). The mean seen-task accuracy after Tasks~1 to 4 is approximately $(34.4,33.9,31.9,28.1)\%$ with enforcement and $(28.9,20.0,24.3,27.7)\%$ without (Fig.~\ref{fig:overview}(f)). Fig.~\ref{fig:overview}(h) shows the pattern per task, with enforcement in the upper triangle. With enforcement, Task~1 accuracy declines modestly from $34.4\%$ to $32.9\%$ after Task~4. Without enforcement, it drops from $28.9\%$ to $16.5\%$ and $13.7\%$ and then partially recovers to $18.8\%$. Retention has a price in the latest task. Current-task accuracy after Tasks~1 to 4 is $(34.4,33.6,29.3,21.4)\%$ with enforcement and $(28.9,23.5,28.7,28.5)\%$ without (Fig.~\ref{fig:overview}(e)). After Task~4, Task~3 accuracy is likewise lower with enforcement ($28.0\%$) than without ($33.3\%$). A rule that protects earlier tasks leaves less room to adapt to the latest one, and a federation sets this balance through the threshold of its policy.

\emph{(Q3) Suitability for the local setting.}
The objective in PersonalizedFL is the accuracy of the personalised model of each client. Enforcement stabilises client-averaged personalised accuracy at approximately $90.5\%$, compared to $83.7\%$ without enforcement (Fig.~\ref{fig:overview}(b)). The personalisation gain, defined as personalised minus reference accuracy, is $53.6$ percentage points with enforcement and $28.0$ without (Fig.~\ref{fig:overview}(d)). Shared-reference accuracy shows the opposite order, with $36.9\%$ and $55.7\%$ (Fig.~\ref{fig:overview}(c)), which is expected. Without enforcement, a higher reference accuracy is thus accompanied by lower personalised accuracy and a smaller gain. Enforcement better preserves client-specific learning, which is the intended behaviour of personalisation.

\emph{Interpretation.}
In all three settings, the objective that the policy stands for improves once the policy is enforced. The results do not show that the selected policies or thresholds are optimal, and they are not a comparison with Byzantine-robust aggregation rules. They show that an agreed rule has a measurable effect on its learning objective when it is enforced. Enforcement in turn requires that compliance can be checked without access to private data, models, and state. The scope of this work is policy compliance and not the improvement of the local-learning and aggregation strategies of each \ac{FL} paradigm.

{\bf EZKL realisations.}
The microbenchmarks quantify what the evidence costs and which party bears the cost. Table~\ref{tab:policy-to-cryptography} reports the timing, circuit, and communication costs of the three EZKL-based realisations over five federated-learning rounds. The mean client-proof generation time is $80.349\pm0.333$~s for VanillaFL, $541.652\pm73.508$~s for ContinualFL, and $640.166\pm17.073$~s for PersonalizedFL. Aggregation-proof generation takes $121.175\pm0.506$~s, $568.800\pm2.764$~s, and $523.948\pm0.765$~s, respectively. The specialised ContinualFL and PersonalizedFL aggregation relations thus require approximately $4.69\times$ and $4.32\times$ the proving time of the VanillaFL relation. Admitting all five client submissions of a round takes $8.900\pm0.043$~s, $28.332\pm0.218$~s, and $28.412\pm0.152$~s. A single client proof verifies in $0.907\pm0.009$~s, $3.911\pm0.051$~s, and $3.917\pm0.039$~s, and an aggregation proof in $1.357\pm0.016$~s, $2.810\pm0.048$~s, and $2.837\pm0.037$~s, measured over 25 client proofs and five aggregation proofs. Client evidence occupies 41.59~KiB for VanillaFL, 2,188.70~KiB for ContinualFL, and 2,189.16~KiB for PersonalizedFL, while aggregation evidence occupies 1,103.54~KiB, 2,176.80~KiB, and 2,174.72~KiB, respectively. The ContinualFL and PersonalizedFL policy relations thus require larger circuits and application payloads (see Sec.~\ref{sec:discussion}).

{\bf Threshold-HE realisation.}
The Threshold-HE implementation results are also reported in Table~\ref{tab:policy-to-cryptography}. This implementation uses native group, range, and threshold-decryption proofs instead of an EZKL circuit. For 3-of-3 decryption, client encryption and evidence generation require $2.491\pm0.019$~s and yield 199.98~KiB of decoded evidence within a 622.54~KiB submission. The aggregator operation, which includes leader aggregation and three helper authorisation/share-generation calls, requires $31.274\pm0.155$~s, while independent saved-round validation requires $7.274\pm0.080$~s. Aggregator-side evidence consists of the three helper-shares and occupies 772.31~KiB. The reported leader-helper exchange involves $9,940.97$~KiB sent and $772.31$~KiB received. Circuit rows and the EZKL capacity parameter $k$ are not applicable to this native realisation, and its different proof mechanism and measurement scope prevent direct comparison with the EZKL results.

{\bf Implications for a deployment.}
The cost of the evidence is asymmetric. With EZKL, a client spends between $80$~s and $640$~s to produce a proof, whereas a verifier or an auditor spends under $4$~s to check one. The party that must demonstrate compliance bears the cost, and checking the demonstration is cheap. The cost also follows the policy. The retention and personalisation policies refer to more private state than the loss policy, and their client proofs take approximately $6.7\times$ and $8.0\times$ as long. A federation can therefore weigh the coverage of a rule against the cost of its evidence.

\begin{table*}[t]
\centering
\caption{AWS microbenchmarks with $28\!\times\!28$ inputs, five clients and five measured rounds, with both policies enabled. Times are mean $\pm$ sample standard deviation; sizes are means. $\uparrow$/$\downarrow$ denote sent/received payloads.}
\label{tab:policy-to-cryptography}
\scriptsize
\resizebox{\textwidth}{!}{%
\begin{tabular}{@{}lllrrrrrrl@{}}
\toprule
Type & Participant & Policy & $k$ & Rows & $N$ & \makecell[c]{Gen./op.\\(s)} & \makecell[c]{Server-side\\validation (s)} & \makecell[c]{Evidence\\(KiB)} & \makecell[c]{Payload\\(KiB)} \\
\midrule
\multirow{2}{*}{Vanilla FL} & Client & $\mathsf{Pred}_{\phi_{\mathsf{loss}}}$ & 20 & 455,104 & 25 & 80.349 $\pm$ 0.333 & -- & 41.59 & $\uparrow$ 104.80/client \\
 & Aggregator & $\mathsf{Pred}_{\psi_{\mathsf{FedAvg}}}$ & 21 & 1,619,920 & 5 & 121.175 $\pm$ 0.506 & 8.900 $\pm$ 0.043 & 1103.54 & \makecell[l]{$\downarrow$ 523.99/cohort\\$\uparrow$ 309.94/cohort} \\
\midrule
\multirow{2}{*}{Continual FL} & Client & $\mathsf{Pred}_{\phi_{\mathsf{EWC}}}$ & 22 & 3,320,727 & 25 & 541.652 $\pm$ 73.508 & -- & 2188.70 & $\uparrow$ 2,283.06/client \\
 & Aggregator & $\mathsf{Pred}_{\psi_{\mathsf{FedAvg}}}$ & 22 & 3,229,170 & 5 & 568.800 $\pm$ 2.764 & 28.332 $\pm$ 0.218 & 2176.80 & \makecell[l]{$\downarrow$ 11,415.32/cohort\\$\uparrow$ 310.47/cohort} \\
\midrule
\multirow{2}{*}{Personalized FL} & Client & $\mathsf{Pred}_{\phi_{\mathsf{pFedMe}}}$ & 22 & 3,102,752 & 25 & 640.166 $\pm$ 17.073 & -- & 2189.16 & $\uparrow$ 2,283.59/client \\
 & Aggregator & $\mathsf{Pred}_{\psi_{\mathsf{pFedMe}}}$ & 22 & 3,230,605 & 5 & 523.948 $\pm$ 0.765 & 28.412 $\pm$ 0.152 & 2174.72 & \makecell[l]{$\downarrow$ 11,417.94/cohort\\$\uparrow$ 311.51/cohort} \\
\midrule
\multirow{2}{*}{Threshold HE FL} & Client & $\mathsf{Pred}_{\phi_{\mathsf{HE\text{-}L2}}}$ & -- & -- & 25 & 2.491 $\pm$ 0.019 & -- & 199.98 & $\uparrow$ 622.54/submission \\
 & Aggregators & $\mathsf{Pred}_{\psi_{\mathsf{HE\text{-}FedAvg}}}$ & -- & -- & 5 & 31.274 $\pm$ 0.155 & 7.274 $\pm$ 0.080 & 772.31 & \makecell[l]{$\uparrow$ 9,940.97/requests\\$\downarrow$ 772.31/shares} \\
\bottomrule
\end{tabular}%
}
\end{table*}

\section{Discussion}\label{sec:discussion}
{\bf Policy cost.}
The cost of the evidence follows how a policy is encoded and not only the number of checks. Table~\ref{tab:policy-to-cryptography} shows that VanillaFL provides the smallest EZKL circuits. The ContinualFL and PersonalizedFL client circuits have about $7\times$ as many rows, because their policies add the EWC and Fisher checks and the checks of pFedMe (Sec.~\ref{sec:Instantiation} and Appendix~\ref{app:additional-instantiations}). Their aggregation circuits have about $2\times$ as many rows, since they verify both the client updates and the shared global model. PersonalizedFL additionally verifies the sampled client set and the server-side mixing parameters. Hence, the row count reflects the encoding, it also leaves room for optimisation. Repeated inputs can be removed, commitments can be reused, and less data can enter the circuit. Such changes reduce circuit size and proving time, in particular when they allow a smaller $k$. Moreover, \proto\ verifies only the training and aggregation properties that a policy selects and not the complete local training. However, if a policy specifies complete local training verification, \proto\ is broad enough to express this. A federation can reuse common policies across \ac{FL} settings and add checks where a setting requires them. It thereby trades verification coverage against proving cost.

{\bf Modularity, interface, and scope.}
The four instantiations use the framework without changes to its core. Threshold-HE adds a cryptographic primitive through the auxiliary component (Def.~\ref{def:aux}) and distributes aggregation across several aggregators. \proto\ is therefore not restricted to a single server, and its roles can be assigned to the nodes of more decentralised settings, where a node acts as client, aggregator, or both. The same separation gives \ac{ML} and applied cryptography a common interface. \ac{ML} researchers specify through policies which properties of training, personalisation, continual learning, or aggregation must be verifiable. Cryptographers optimise the commitments, proof systems, and secure-aggregation mechanisms that enforce these policies. An improvement on either side carries over without redesigning the framework. Since our definitions fix no particular commitment scheme or proof system, compatible components can also be replaced, including by post-quantum alternatives. Finally, the guarantees depend on the selected policies and on the threat model of Sec.~\ref{sec:threat-model}. Client policies decide which updates are admitted. Aggregation policies require every recorded contribution to be processed exactly once with the prescribed rule and weights, relative to the recorded admitted set. Additional properties such as confidentiality of individual contributions and their integrity and origin require an auxiliary mechanism which \proto\ can express through $\Aux$. 

\section{Related Work}\label{sec:related-work}
Prior work on verifiable \ac{FL} mostly asks whether a contribution or an aggregate is correct or harmless. \proto\ asks whether a federation can demonstrate that the rules it agreed on were followed, which is the evidence that the expectations in Table~\ref{tab:regulatory-mapping} call for. A recent systematisation organises integrity claims in privacy-preserving \ac{FL} by protocol stage~\cite{rizzini2026sok}. We, however contrast each line of work with Requirements~\textbf{(R1)} to \textbf{(R5)}.

\noindent\textbf{Robust aggregation and input validation.}
Byzantine-robust aggregation filters statistically unusual updates at a trusted server and leaves no evidence~\cite{blanchard2017krum,cao2021fltrust}. Verifiable input validation removes this trust for the client side. EIFFeL checks programmable validity predicates on secret-shared updates~\cite{chowdhury2022eiffel}. RoFL, ACORN, and RiseFL prove norm bounds on updates that stay hidden under secure aggregation~\cite{lycklama2023rofl,bell2023acorn,zhu2024risefl}, and BPFL and NoV encode magnitude, distance, or direction conditions~\cite{nie2024bpfl,xing2024novandalism}. These systems address Requirement~\textbf{(R2)} for public conditions on the submitted update that bound the influence of malicious clients. The policies behind (Q1) to (Q3) instead refer to committed private data, a private reference model, or a private personalised model, which a condition on the update alone cannot express. Except for NoV, these systems also give no evidence that the aggregate follows an aggregation policy~\textbf{(R3)}. In \proto, a norm bound is one policy among others (Appendix~\ref{app:threshold-he-instantiation}).

\noindent\textbf{Verifiable training.}
A second line of work proves how a model was computed. Proofs of training cover every step of local training~\cite{garg2023zkpot,abbaszadeh2024kaizen}, and VeriDP proves the steps of DP-SGD~\cite{abdolmaleki2026veridp}. In \ac{FL}, PZKP-FL encodes local training and aggregation as zk-SNARK circuits~\cite{xing2023pzkpfl}, Federify does so for Gaussian Naive Bayes with threshold ElGamal and a blockchain~\cite{keshavarzkalhori2024federify}, and VPFL verifies client updates and the aggregate with zero-knowledge proofs and commitments~\cite{ma2024vpfl}. These systems cover both sides of the protocol, but each verifies one fixed algorithm. The prover pays for every step, and a new algorithm needs a new construction. \proto\ proves the policy rather than the computation. A loss condition, a retention bound, or a stationarity condition certifies a property of the result at a cost that follows the policy, and complete training remains expressible as a policy when a federation requires it.

\noindent\textbf{Trusted hardware.}
VerifiableFL/ExclaveFL builds an attested dataflow over training and aggregation tasks~\cite{guo2024verifiablefl}. Sentinel attests the local training process of each client to the server~\cite{zhang2025sentinel}, and Takagi and Hasegawa combine a trusted execution environment, secure aggregation, and client auditing to verify a stateful private-\ac{FL} server~\cite{takagi2026securing}. Such evidence relies on trusted hardware and its attestation infrastructure, and it ties verification to one execution platform. \proto\ relies on standard cryptographic assumptions instead, and the two approaches are complementary.

\noindent\textbf{Verifiable aggregation.}
VerifyNet lets clients check the algebraic consistency of a returned aggregate with homomorphic hashing~\cite{xu2020verifynet}, and zkFL proves with Halo2 that the aggregate is computed from committed client contributions~\cite{wang2025zkfl}. They protect the clients against a deviating server but leave the client side unchecked. The aggregation policy of \proto\ covers more than the sum. It fixes the recorded admitted set, the weights, the aggregation rule, and the model update, which may be a weighted average, a personalisation mixing step, or a threshold decryption (Sec.~\ref{sec:Instantiation} and Appendix~\ref{app:additional-instantiations}).

\noindent\textbf{Policy compliance and accountability.}
Outside \ac{FL}, privacy rules have been formalised and checked against systems, from contextual-integrity norms for laws such as HIPAA~\cite{barth2006ci} to automated compliance checks of a production data-analytics pipeline~\cite{sen2014bootstrapping}. Accountable algorithms use commitments and zero-knowledge proofs to show that a decision followed an announced policy without revealing secret inputs~\cite{kroll2017accountable}, and recent work proves differential privacy or fairness of a centrally trained model in zero knowledge~\cite{shamsabadi2024confidentialdpproof,yadav2024fairproof}. \proto\ brings this idea to federated training, where the rules concern many parties, data that never leave the clients, and state retained across rounds.

\noindent\textbf{Positioning.}
Prior systems verify one condition, one computation, or one aggregate within one workflow, and their guarantees are tied to the cryptographic technique they use. EIFFeL comes closest with programmable predicates, but these remain conditions on the submitted update within one secure-aggregation protocol. To our knowledge, none of the systems above treats both client and aggregation rules as policy relations over private data and retained state, configures them for the selected \ac{FL} type and topology~\textbf{(R1)}, and separates their semantics from the mechanisms that check them~\textbf{(R5)}. Our instantiations apply this abstraction to vanilla, continual, personalised, and threshold-encrypted \ac{FL} with several aggregators. The guarantees remain relative to the selected policy and to the recorded admitted set (Sec.~\ref{sec:threat-model}).

\section{Conclusion}\label{sec:conclusion}
In this paper, we introduced \proto, a framework that allows the definition of policies based on which policy-compliant \ac{FL} settings can be instantiated. To this end, we provided a thorough formalisation of \proto\ by specifying the necessary components and abstractions to support \emph{modularity} and \emph{crypto-agility}. In addition, we demonstrated the versatility of the framework through several \ac{FL} instantiations, namely Vanilla FL, Continual FL, Personalised FL, and Threshold HE-based FL. In tandem, we provided a thorough security analysis showing that \proto\ achieves three main security properties: (i) update admissibility, (ii) aggregation correctness relative to the recorded set, and (iii) policy-relative privacy.
Finally, to demonstrate that \proto\ is not only a theoretical construction, we implemented proof-of-concept realisations of all presented instantiations, benchmarked their performance, and reported the resulting computational and verification overheads.

\section*{Acknowledgments}
This work is supported by the European Network of Excellence dAIEDGE under Grant Agreement Nr. 101120726 and by the Flemish Government through the Cybersecurity Research Program with grant number: VOEWICS02.


\bibliographystyle{ieeetr}
\bibliography{biblio_2,proofs}

\appendices

\section{Additional Framework Instantiations}\label{app:additional-instantiations}
Section~\ref{sec:Instantiation} answers (Q1) with VanillaFL, and this appendix answers (Q2) to (Q4). Clients and aggregators build, prove, and verify $\mu_{i,t}=(i,t,x_{i,t},\pi_{i,t})$ and $\nu_t=(\ell,t,x_t^{\mathsf{agg}},\pi_t^{\mathsf{agg}})$ as in Algorithms~\ref{alg:client} and~\ref{alg:aggregator}, with the relations given below.

\subsection{Continual Federated Learning}
\begin{question}{(Q2)}
   When the model is extended to new tasks, does each contribution respect the agreed bound on forgetting earlier ones? This supports consistent performance over the lifecycle~\cite[Art.~15(1)]{aiact2024} and verified model changes~\cite{pccp2023principles}, although the state that the check refers to is private to each client. 
\end{question}

We instantiate $\mathsf{Type}=\mathsf{ContinualFL}$ with clients $\{\client{i}\}_{i\in I}$ and one aggregator $\agg{1}$, using Elastic Weight Consolidation (EWC)~\cite{Kim2024cl,ewc-paper}. Client $\client{i}$ retains the state $s_{i,t-1}=(\theta^{\mathsf{ref}}_{i,t-1},F_{i,t-1})$ of its last accepted participation, which need not be the preceding round. Here $\theta^{\mathsf{ref}}_{i,t-1}$ is its most recent accepted local model, $F_{i,t-1}$ is the diagonal Fisher information at that model, and $s_{i,0}=(\theta_0,\mathbf{0})$. The commitment $c^s_{i,t-1}$ to $s_{i,t-1}$, with private opening $r^s_{i,t-1}$, is retained from the client's previous admitted message or created before the first round. In round $t$, the client holds the private task dataset $D_{i,t}$ of size $n_{i,t}=|D_{i,t}|$. It computes $u_{i,t}\gets\LocalTrain(\mathsf{ContinualFL},\theta_t,D_{i,t},\Params)$ and $\theta_{i,t}=\theta_t+u_{i,t}$ by optimising the EWC objective $\mathsf{Loss}(\theta,D_{i,t})+\tfrac{\lambda}{2}\sum_{j=1}^{d}F_{i,t-1,j}(\theta_j-\theta^{\mathsf{ref}}_{i,t-1,j})^2$. It then computes the diagonal empirical Fisher information $F_{i,t}=\mathsf{DiagFisher}(\theta_{i,t},D_{i,t},\Params)$, with $F_{i,t,j}=\tfrac{1}{n_{i,t}}\sum_{(x,y)\in D_{i,t}}\bigl(\partial_{\theta_j}\log p_{\theta}(y\mid x)\big|_{\theta=\theta_{i,t}}\bigr)^2$ for the estimator fixed in $\Params$, and forms its next state $s_{i,t}=(\theta_{i,t},F_{i,t})$.

\emph{Client policy.}
The policy $\phi_{\mathsf{EWC}}\in\Phi$ commits to $M_{i,t}=(D_{i,t},u_{i,t},s_{i,t})$. The values $u_{i,t}$, $n_{i,t}$, and $c^s_{i,t-1}$ are public, whereas $D_{i,t}$, $s_{i,t-1}$, $s_{i,t}$, and all openings remain private. No auxiliary mechanism is required.
The client predicate is

\[
\boxed{
\scriptsize
\begin{aligned}
&\mathsf{Pred}_{\phi_{\mathsf{EWC}}}(x_{i,t},w_{i,t})=1
\Longleftrightarrow
c^s_{i,t-1}\text{ binds }\bigl(\theta^{\mathsf{ref}}_{i,t-1},F_{i,t-1}\bigr)
\\
&\quad{}\land C_{i,t}\text{ binds }(D_{i,t},u_{i,t},s_{i,t})
\land n_{i,t}=|D_{i,t}|
\land \theta_{i,t}=\theta_t+u_{i,t}
\\
&\quad{}\land E_{i,t}=\tfrac{\lambda}{2}\textstyle\sum_{j=1}^{d}F_{i,t-1,j}\bigl(\theta_{i,t,j}-\theta^{\mathsf{ref}}_{i,t-1,j}\bigr)^2
\land E_{i,t}\leq\tau_{\mathsf{EWC}}
\\
&\quad{}\land F_{i,t}=\mathsf{DiagFisher}(\theta_{i,t},D_{i,t},\Params)
\land s_{i,t}=(\theta_{i,t},F_{i,t}).
\end{aligned}
}
\]

The relation certifies that the local model satisfies the EWC bound relative to the previously committed state and that the retained Fisher vector is computed at the resulting local model. It does not certify the execution of $\LocalTrain$. This answers (Q2). A norm bound on the update cannot express this rule, because the bound is anchored at a private reference model and weighted by private Fisher information.

\emph{Aggregation policy.}
ContinualFL reuses $\psi_{\mathsf{FedAvg}}$ from Sec.~\ref{sec:Instantiation} with the extracted values $V_t[i]=(u_{i,t},n_{i,t})$. The next-state commitment is not aggregated. It remains in the admitted message and becomes the client's reference-state commitment for its next participation.

\subsection{Personalised Federated Learning}\label{app:personalised-instantiation}
\begin{question}{(Q3)}
   Does a hospital that adapts the model to its own population still derive its shared contribution from the agreed personalisation objective? Personalisation serves the specific setting of use~\cite[Art.~10(4)]{aiact2024}, since models can transfer poorly between sites~\cite{zech2018variable}, but the personalised model never leaves the hospital.  
\end{question}

We instantiate $\mathsf{Type}=\mathsf{PersonalizedFL}$ with clients $\{\client{i}\}_{i\in I}$ and one aggregator $\agg{1}$ using pFedMe~\cite{dinh2020pfedme}. In round $t$, client $\client{i}$ holds a private dataset $D_i$ and computes local reference models $\omega_{i,t,0},\ldots,\omega_{i,t,R}$ with $\omega_{i,t,0}=\theta_t$. In local round $r$, it computes a private personalised model $\theta^{\mathsf{pers}}_{i,t,r}$ for the proximal objective $\widetilde h_{i,t,r}(v)=\mathsf{Loss}_{\mathsf{circ}}(v,D_i)+\tfrac{\lambda_{\mathsf M}}{2}\|v-\omega_{i,t,r}\|_2^2$, where $\mathsf{Loss}_{\mathsf{circ}}$ is the circuit-compatible loss fixed by $\Params$. It then updates the reference model by the pFedMe recurrence in the predicate below. The public contribution is $u_{i,t}=\omega_{i,t,R}-\theta_t$, and the personalised models remain private.

\emph{Client policy.}
The policy $\phi_{\mathsf{pFedMe}}\in\Phi$ commits to the dataset, the private personalised and reference-model trajectory, and the public update $u_{i,t}$. It requires each personalised point to satisfy an approximate first-order condition and each transition to follow the pFedMe recurrence up to the fixed-point tolerance $\simeq_{\varepsilon_{\mathsf{rec}}}$ prescribed by $\Params$.
The client predicate is
\[
\boxed{
\scriptsize
\begin{aligned}
&\mathsf{Pred}_{\phi_{\mathsf{pFedMe}}}(x_{i,t},w_{i,t})=1
\Longleftrightarrow
\\
&\quad C_{i,t}\text{ binds }\bigl(D_i,\{\omega_{i,t,r}\}_{r=0}^{R},\{\theta^{\mathsf{pers}}_{i,t,r}\}_{r=0}^{R-1},u_{i,t}\bigr)
\land \omega_{i,t,0}=\theta_t
\\
&\quad{}\land \textstyle\bigwedge_{r=0}^{R-1}\bigl\|\nabla_v\widetilde h_{i,t,r}\bigl(\theta^{\mathsf{pers}}_{i,t,r}\bigr)\bigr\|_2^2\leq\tau_{\mathsf{prox}}
\land u_{i,t}=\omega_{i,t,R}-\theta_t
\\
&\quad{}\land \textstyle\bigwedge_{r=0}^{R-1}\omega_{i,t,r+1}\simeq_{\varepsilon_{\mathsf{rec}}}\omega_{i,t,r}-\eta\lambda_{\mathsf M}\bigl(\omega_{i,t,r}-\theta^{\mathsf{pers}}_{i,t,r}\bigr).
\end{aligned}
}
\]

The relation certifies that the submitted update follows from personalised points that satisfy the proximal-stationarity bound and from the pFedMe recurrence. It does not require a particular optimiser or an exact minimiser.

\emph{Aggregation policy.}
The extraction interface returns $V_t[i]=u_{i,t}$. For the cohort $S_t$ prescribed for round $t$, the policy $\psi_{\mathsf{pFedMe}}\in\Phi$ requires $I_t^{\mathsf{adm}}=S_t$ and applies the pFedMe mixing rule in the fixed-point arithmetic of $\Params$. Its predicate is
\[
\boxed{
\scriptsize
\begin{aligned}
&\mathsf{Pred}_{\psi_{\mathsf{pFedMe}}}(x_t^{\mathsf{agg}},w_t^{\mathsf{agg}})=1
\Longleftrightarrow
I_t^{\mathsf{adm}}=S_t
\land V_t[i]=u_{i,t}\ \forall i\in S_t
\\
&\quad{}\land \Delta_t=\tfrac{\beta}{|S_t|}\textstyle\sum_{i\in S_t}V_t[i]
\land \theta_{t+1}=\theta_t+\Delta_t .
\end{aligned}
}
\]

The two proofs answer (Q3). Each shared update follows from personalised points that satisfy the agreed objective up to the stated tolerance, and the server mixed exactly the prescribed cohort with the prescribed $\beta$. Our evaluation uses one local round, a full local batch, and full-cohort participation, which are choices of the realisation. The policy certifies the round-local pFedMe relation, not personalised accuracy or cross-round continuity of a personalised model.

\subsection{Threshold Homomorphic Federated Learning}\label{app:threshold-he-instantiation}
\begin{question}{(Q4)}
    If no single server may read a contribution, is each contribution still well formed, and is the aggregate correct? Shared updates can leak training data~\cite{zhu2019dlg}, so the UK ICO recommends combining \ac{FL} with homomorphic encryption~\cite{ico2023pets}, which removes the server's view of the contributions. Here, the client policy is a norm bound, one policy among others in \proto.
\end{question}

We instantiate $\mathsf{Type}=\mathsf{Threshold\text{-}HE}$, following threshold-encrypted FL constructions~\cite{tHEZhu}, with clients $\{\client{i}\}_{i\in I}$ and aggregators $\{\agg{\ell}\}_{\ell\in J}$. The auxiliary component $\mathsf{HE}\in\Aux$ is an additively homomorphic threshold-encryption scheme. Its setup in Algorithm~\ref{alg:generic-framework-setup-compact} produces the public key $\mathsf{pk}_{\mathsf{HE}}$ and the decryption shares of the aggregators. Our prototype uses 3-of-3 threshold ElGamal~\cite{elgamal85} over Ristretto255~\cite{ristretto255} (\texttt{elastic-elgamal}~0.3.1~\cite{elasticElGamal}), ElGamal--Pedersen consistency proofs and Bulletproofs~\cite{bulletproofs} for the range and norm checks, Chaum--Pedersen proofs~\cite{chaum-pedersen} for decryption shares, and a trusted dealer instead of distributed key generation. In round $t$, client $\client{i}$ computes $u_{i,t}\gets\LocalTrain(\mathsf{Threshold\text{-}HE},\theta_t,D_{i,t},\Params)$ and its integer encoding $q_{i,t}=\Encode(u_{i,t},\Params)$. The public parameters fix a coordinate bound $Q$ and a squared-norm bound $B_q^2$.

\emph{Client policy.}
The client commits to $M_{i,t}=(q_{i,t})$ and encrypts the same vector as $E_{i,t}=\mathsf{Enc}(\mathsf{pk}_{\mathsf{HE}},q_{i,t};\eta_{i,t})$. The ciphertext $E_{i,t}$ is public auxiliary output, and the randomness $\eta_{i,t}$ remains private. The policy $\phi_{\mathsf{HE\text{-}L2}}\in\Phi$ binds both to the same vector and bounds it.
The client predicate is
\[
\boxed{
\scriptsize
\begin{aligned}
&\mathsf{Pred}_{\phi_{\mathsf{HE\text{-}L2}}}(x_{i,t},w_{i,t})=1
\Longleftrightarrow
C_{i,t}\text{ binds }q_{i,t}
\\
&\quad{}\land E_{i,t}=\mathsf{Enc}(\mathsf{pk}_{\mathsf{HE}},q_{i,t};\eta_{i,t})
\\
&\quad{}\land |q_{i,t,j}|\leq Q\ \text{for every }j
\land \|q_{i,t}\|_2^2\leq B_q^2 .
\end{aligned}
}
\]

The relation binds the ciphertext to the committed encoding and certifies both bounds without revealing $q_{i,t}$.

\emph{Aggregation policy.}
The extraction interface returns $V_t[i]=(E_{i,t},n_i)$, where $n_i$ is the public weight fixed in $\Params$, and $N_t=\sum_{i\in I_t^{\mathsf{adm}}}n_i$. The aggregators compute $E_t^{\mathsf{sum}}=\mathsf{Eval}(\mathsf{pk}_{\mathsf{HE}},\{(n_i,E_{i,t})\}_{i\in I_t^{\mathsf{adm}}})$, and the combination of their decryption shares yields $s_t=\sum_{i\in I_t^{\mathsf{adm}}}n_iq_{i,t}$. The next model is $\theta_{t+1}=\theta_t+\mathsf{DecodeAvg}(s_t,N_t,\Params)$. The policy $\psi_{\mathsf{HE\text{-}FedAvg}}\in\Phi$ requires that $V_t$ contains exactly the ciphertexts and weights extracted from $\mathsf{Adm}_t$, that $E_t^{\mathsf{sum}}$ is the prescribed evaluation, that the decryption shares supplied through the auxiliary component are valid, and that decoding yields the published $\theta_{t+1}$.

This answers (Q4). Every admitted ciphertext contains a well-formed update within the agreed bounds, and the published model follows from these ciphertexts through the prescribed aggregation and a valid threshold decryption. If fewer aggregators than the decryption threshold collude, no aggregator sees an individual update. Theorems~\ref{thm:update-adm} to~\ref{thm:policy-privacy} treat ciphertexts as public and do not cover this confidentiality. It follows from the semantic security of ElGamal under the decisional Diffie--Hellman assumption~\cite{tsiounis1998elgamal}, honest key generation by the trusted dealer, and the zero-knowledge and hiding properties used in Theorem~\ref{thm:policy-privacy}.

\section{Security Definitions and Proofs}
\label{app:security-proofs}
This appendix gives condensed definitions, games, and proofs for the theorems of Sec.~\ref{sec:security-analysis}, with the games in Fig.~\ref{fig:games}. The artefact (Appendix~\ref{app:open-science}) contains the detailed version, which also proves the extensions to many rounds and messages used in Sec.~\ref{sec:security-analysis}.

\subsection{Cryptographic Assumptions}\label{app:crypto-notions}
We use the standard notions for non-interactive proofs and commitments~\cite{blum1988nizk,feige1999multiple,bellare1992pok,pedersen1991commitment} and fix the notation below. Let $\mathcal{R}$ be a relation supported by $\Pi$ with language $\mathcal{L}_{\mathcal{R}}$, let $\mathcal{B}$ be a $\mathsf{PPT}$ adversary, and let $\mathsf{pp}_{\Pi}\gets\mathsf{Setup}_{\Pi}(1^\lambda)$ and $\mathsf{pp}_{\mathsf{CS}}\gets\mathsf{Setup}_{\mathsf{CS}}(1^\lambda)$.

\emph{Soundness and knowledge soundness.}
$\operatorname{Adv}^{\mathsf{snd}}_{\Pi,\mathcal{B}}(\lambda)$ is the probability that $\mathcal{B}(\mathsf{pp}_{\Pi})$ outputs $(x,\pi)$ with $x\notin\mathcal{L}_{\mathcal{R}}$ and $\Pi.\Verify_{\mathcal{R}}(\mathsf{pp}_{\Pi},x,\pi)=1$. $\operatorname{Adv}^{\mathsf{k\text{-}snd}}_{\Pi,\mathcal{B}}(\lambda)$ is the probability that the proof verifies while an extractor $\mathcal{E}_{\mathcal{B}}$, run on the same inputs and random coins, outputs $w$ with $(x,w)\notin\mathcal{R}$. $\Pi$ is knowledge-sound if every $\mathcal{B}$ has an $\mathcal{E}_{\mathcal{B}}$ that makes this probability negligible~\cite{bellare1992pok}. A statement $x\notin\mathcal{L}_{\mathcal{R}}$ has no witness, so $\operatorname{Adv}^{\mathsf{snd}}_{\Pi,\mathcal{B}}(\lambda)\leq\operatorname{Adv}^{\mathsf{k\text{-}snd}}_{\Pi,\mathcal{B}}(\lambda)$. Theorems~\ref{thm:update-adm} and~\ref{thm:agg-correct} use $\operatorname{Adv}^{\mathsf{snd}}$.

\emph{Zero knowledge.}
$\operatorname{Adv}^{\mathsf{zk}}_{\Pi,\mathcal{B}}(\lambda)$ is the advantage of $\mathcal{B}$ in distinguishing real proofs $\Pi.\Prove_{\mathcal{R}}(\mathsf{pp}_{\Pi},x,w)$ from simulated proofs $\mathsf{SimProve}_{\mathcal{R}}(\tau,x)$ under $(\widetilde{\mathsf{pp}}_{\Pi},\tau)\gets\mathsf{SimSetup}(1^\lambda)$, for polynomially many adaptively chosen pairs $(x,w)\in\mathcal{R}$. This is multi-theorem zero knowledge~\cite{feige1999multiple}.

\emph{Hiding and binding.}
$\operatorname{Adv}^{\mathsf{hid}}_{\mathsf{CS},\mathcal{B}}(\lambda)=|2\Pr[b'=b]-1|$ for a guess $b'$ of a random bit $b$ from $C_b$, where $(C_b,R_b)\gets\Commit(\mathsf{pp}_{\mathsf{CS}},M_b)$ for valid messages $M_0,M_1$ chosen by $\mathcal{B}$. $\operatorname{Adv}^{\mathsf{bind}}_{\mathsf{CS},\mathcal{B}}(\lambda)$ is the probability that $\mathcal{B}$ opens one commitment to two different messages.

Auxiliary outputs enter the statements as public values $Z^{\mathsf{pub}}$, and the policy relations encode their consistency, so the proofs need no assumption on $\Aux$. A guarantee that an auxiliary primitive adds, such as the confidentiality of encrypted updates, relies on the assumptions of the instantiation (Def.~\ref{def:aux}).

\subsection{Security Game Notions and Notation}\label{app:security-game-notions}
The games differentiate \emph{verification} from \emph{semantic validity}. Verification methods $\Framework.\mathsf{Vfy}(\cdot)$ are the efficient checks that honest parties perform during the protocol. Validity predicates $\Framework.\mathsf{Valid}(\cdot)$ state the policy conditions that an accepted message is meant to satisfy. They are used only in the definitions and need not be efficiently decidable. $\Framework.\mathsf{Setup}$ denotes the generic setup of Algorithm~\ref{alg:generic-framework-setup-compact}, with the framework instance and participant sets suppressed from the notation.

\noindent\textbf{Client messages.}
Fix a round $t$ and a client policy $\phi$, and let $\kappa_t^\phi=(t,\phi,\theta_t,\Params)$ denote the active public context. 
A client message $\mu_{i,t}=(\hat{i},\hat{t},x_{i,t},\pi_{i,t})$ is received from client $i$ over an authenticated channel. 
The channel therefore determines identity $i$, not the message. 
Client verification and client validity share the deterministic checks $\mathsf{Chk}$, where $\mathsf{Chk}(\kappa_t^\phi,i,\mu_{i,t})=1$ if and only if $\hat{i}=i$, $\hat{t}=t$, $\mathsf{Ctx}(x_{i,t})=(i,t,\phi,\theta_t,\Params)$, and $\mathsf{ExtractAgg}_{\mathsf{Type},\phi}(x_{i,t})\neq\bot$. They differ only in the final condition.
Client verification $\Framework.\mathsf{VfyClient}(\kappa_t^\phi,i,\mu_{i,t})$ is $1$ if and only if $\mathsf{Chk}(\kappa_t^\phi,i,\mu_{i,t})=1$ and $\Pi.\Verify_{\mathcal{R}_\phi}(\ppnizk,x_{i,t},\pi_{i,t})=1$. Client validity $\Framework.\mathsf{ValidClient}(\kappa_t^\phi,i,\mu_{i,t})$ is $1$ if and only if $\mathsf{Chk}(\kappa_t^\phi,i,\mu_{i,t})=1$ and $x_{i,t}\in\mathcal{L}_{\phi}$, where $x\in\mathcal{L}_{\phi}$ if and only if $(x,w)\in\mathcal{R}_{\phi}$ for some witness $w$.

\noindent\textbf{Records and aggregation messages.}
Let $\mathsf{Rec}_t=\bigl((i_1,\mu_{i_1,t}),\ldots,(i_m,\mu_{i_m,t})\bigr)$ denote the public admitted-message record $\mathsf{Adm}_t$ of round $t$ (Algorithm~\ref{alg:aggregator}), with each message paired with its authenticated sender identity $i_j$. 
$\Framework.\mathsf{ValidRec}(\kappa_t^\phi,\mathsf{Rec}_t)=1$ if and only if the identifiers in $\mathsf{Rec}_t$ are pairwise distinct and $\Framework.\mathsf{ValidClient}(\kappa_t^\phi,i,\mu_{i,t})=1$ for every $(i,\mu_{i,t})\in\mathsf{Rec}_t$. 
We write $\Framework.\mathsf{AggOut}_{\mathsf{Type},\psi}(\theta_t,\mathsf{Rec}_t,\Params)$ for the deterministic output obtained by extracting exactly one aggregation value from every recorded message, applying the aggregation rule and weights prescribed by $\mathsf{Type}$ and $\psi$, and applying the prescribed model-update rule.
An aggregation message has the form $\nu_t=(\ell,\hat{t},x_t^{\mathsf{agg}},\pi_t^{\mathsf{agg}})$
with
\(
x_t^{\mathsf{agg}}=(\hat{\ell},\tilde{t},\hat{\phi},\hat{\psi},\hat{\theta},\widehat{\mathsf{Rec}}_t,\theta_{t+1},Z_t^{\mathsf{pub}},\widehat{\Params}).
\)
We write $\Framework.\mathsf{VfyAgg}(\kappa_t^\phi,\psi,\mathsf{Rec}_t,\nu_t)=1$ if and only if four checks succeed.
(1)~The message and statement match the active context, namely $\hat{t}=\tilde{t}=t$, $\hat{\ell}=\ell$, $\hat{\phi}=\phi$, $\hat{\psi}=\psi$, $\hat{\theta}=\theta_t$, $\widehat{\Params}=\Params$, and $\widehat{\mathsf{Rec}}_t=\mathsf{Rec}_t$. 
(2)~The identifiers in $\mathsf{Rec}_t$ are pairwise distinct. 
(3)~$\Framework.\mathsf{VfyClient}(\kappa_t^\phi,i,\mu_{i,t})=1$ for every $(i,\mu_{i,t})\in\mathsf{Rec}_t$. 
(4)~$\Pi.\Verify_{\mathcal{R}_{\psi}}(\ppnizk,x_t^{\mathsf{agg}},\pi_t^{\mathsf{agg}})=1$. The corresponding validity predicate $\Framework.\mathsf{ValidAgg}(\kappa_t^\phi,\psi,\mathsf{Rec}_t,\nu_t)=1$ holds if and only if $\Framework.\mathsf{ValidRec}(\kappa_t^\phi,\mathsf{Rec}_t)=1$ and $x_t^{\mathsf{agg}}\in\mathcal{L}_{\psi}$.
Since $\mathcal{R}_{\psi}$ enforces the prescribed participation, extraction, weighting, aggregation, and model update, $x_t^{\mathsf{agg}}\in\mathcal{L}_{\psi}$ implies $\theta_{t+1}=\Framework.\mathsf{AggOut}_{\mathsf{Type},\psi}(\theta_t,\mathsf{Rec}_t,\Params)$, using the correctness of any auxiliary mechanism, such as homomorphic evaluation and threshold decryption.

\noindent\textbf{Privacy notation.}
Let $s$ denote the private client input and local state required by the selected FL type, client policy, and realisation. We write $(x,w)\gets\Framework.\mathsf{BuildClient}(\kappa_t^\phi,i,s)$ for running Algorithm~\ref{alg:client} up to the formation of the statement and witness, before the relation check and proof generation. 
For a client statement $ x_{i,t}=(i,t,\phi,\theta_t,C_{i,t},Z_{i,t}^{\mathsf{pub}},\Params),$
we define $\Framework.\mathsf{CommonCtx}(x_{i,t})=(i,t,\phi,\theta_t,Z_{i,t}^{\mathsf{pub}},\Params).$
These are the public fields that must be identical in the two privacy worlds. The commitment $C_{i,t}$ is excluded from them but remains part of the statement.

\subsection{Update Admissibility}\label{app:game-adm}
\begin{figure*}[t]
\centering
\scriptsize
\begin{minipage}[t]{0.485\textwidth}
\vspace{0pt}
\GameBox{(a) Game $\mathsf{Game}^{\mathsf{adm}}_{\Framework,\mathcal{A}_1}(1^\lambda)$}{%
$\triangleright$ \textbf{Setup.} Challenger $\mathcal{C}$ computes
\(
(\mathsf{pp}_{\Framework},
 \{\mathsf{st}^{\mathsf{setup}}_P\}_{P\in\mathcal P})
\gets
\Framework.\mathsf{Setup}(1^\lambda).
\)
$\mathcal{C}$ gives
$(\mathsf{pp}_{\Framework},\kappa_t^\phi)$
and the prescribed setup state of the corrupted client to $\mathcal{A}_1$.

$\triangleright$ \textbf{Adversary output.} The adversary, acting as corrupted client $i$, outputs
\(
\mu_{i,t}
\gets
\mathcal{A}_1
(\mathsf{pp}_{\Framework},\kappa_t^\phi).
\)

\textbf{Return} $1$ if and only if
\[
\Framework.\mathsf{VfyClient}
(\kappa_t^\phi,i,\mu_{i,t})=1
\,\land\,
\Framework.\mathsf{ValidClient}
(\kappa_t^\phi,i,\mu_{i,t})=0.
\]
}

\vspace{4pt}
\GameBox{(b) Game $\mathsf{Game}^{\mathsf{agg}}_{\Framework,\mathcal{A}_2}(1^\lambda)$}{%
$\triangleright$ \textbf{Setup.} Challenger $\mathcal{C}$ computes
\(
(\mathsf{pp}_{\Framework},
 \{\mathsf{st}^{\mathsf{setup}}_P\}_{P\in\mathcal P})
\gets
\Framework.\mathsf{Setup}(1^\lambda).
\)
$\mathcal{C}$ gives
\(
(\mathsf{pp}_{\Framework},\kappa_t^\phi,\psi)
\)
and the prescribed setup state of the corrupted aggregator to
$\mathcal{A}_2$.

$\triangleright$ \textbf{Adversary output.} The adversarial aggregator outputs
\(
(\mathsf{Rec}_t,\nu_t)
\gets
\mathcal{A}_2
(\mathsf{pp}_{\Framework},\kappa_t^\phi,\psi).
\)

\textbf{Return} $1$ if and only if
\[
\begin{aligned}
&\Framework.\mathsf{VfyAgg}
(\kappa_t^\phi,\psi,\mathsf{Rec}_t,\nu_t)=1
\\
&\quad {}\land
\Framework.\mathsf{ValidAgg}
(\kappa_t^\phi,\psi,\mathsf{Rec}_t,\nu_t)=0.
\end{aligned}
\]
}
\end{minipage}\hfill
\begin{minipage}[t]{0.485\textwidth}
\vspace{0pt}
\GameBox{(c) Game $\mathsf{Game}^{\mathsf{priv}}_{\Framework,\mathcal{A}_2}(1^\lambda,b)$}{%
$\triangleright$ \textbf{Setup.} Challenger $\mathcal{C}$ computes
\(
(\mathsf{pp}_{\Framework},
 \{\mathsf{st}^{\mathsf{setup}}_P\}_{P\in\mathcal P})
\gets
\Framework.\mathsf{Setup}(1^\lambda).
\)
$\mathcal{C}$ gives
\(
(\mathsf{pp}_{\Framework},\kappa_t^\phi)
\)
and the prescribed setup state of the corrupted participant to $\mathcal{A}_{2,1}$.

$\triangleright$ \textbf{Challenge selection.} The first-stage adversary chooses a client $i$ and two private inputs $s_0,s_1$, each holding the client data and local state required by $\phi$. It outputs
\(
(i,s_0,s_1,\mathsf{st})
\gets
\mathcal{A}_{2,1}
(\mathsf{pp}_{\Framework},\kappa_t^\phi).
\)

$\triangleright$ \textbf{Challenge construction.} For each $j\in\{0,1\}$, $\mathcal{C}$ builds the statement and witness that client $i$ would form from $s_j$,
\(
(x_j,w_j)
\gets
\Framework.\mathsf{BuildClient}
(\kappa_t^\phi,i,s_j).
\)

$\triangleright$ \textbf{Admissibility.} Both inputs must be valid under $\phi$ and must agree on every permitted public value. If
\(
\Framework.\mathsf{CommonCtx}(x_0)
\neq
\Framework.\mathsf{CommonCtx}(x_1)
\)
or
\(
(x_j,w_j)\notin\mathcal{R}_{\phi}
\)
for some $j\in\{0,1\}$,
\textbf{Return} $0$. Otherwise, $x_0$ and $x_1$ differ only in their commitments.

$\triangleright$ \textbf{Challenge.} $\mathcal{C}$ proves the statement selected by the bit $b$,
\(
\pi_b
\gets
\Pi.\Prove_{\mathcal{R}_\phi}(\ppnizk,x_b,w_b),
\)
and gives $(x_b,\pi_b)$ to $\mathcal{A}_{2,2}$.

$\triangleright$ \textbf{Adversary output.} Besides its state $\mathsf{st}$, the second-stage adversary sees only the statement and the proof. It outputs a guess
\(
b'
\gets
\mathcal{A}_{2,2}
(\mathsf{st},x_b,\pi_b).
\)

\textbf{Return} $b'$.
}
\end{minipage}
\caption{Security games. (a)~Update admissibility. (b)~Aggregation correctness. (c)~Policy-relative privacy.}
\label{fig:games}
\end{figure*}

The game in Fig.~\ref{fig:games}(a) formalises Property~\textbf{(P1)}. A corrupted client wins if its message passes client verification but is not valid under the active client policy, in a round $t$ of an honest execution with active context $\kappa_t^\phi$.

\begin{definition}[Update admissibility]
The update-admissibility advantage of $\mathcal{A}_1$ is $\operatorname{Adv}^{\mathsf{adm}}_{\Framework,\mathcal{A}_1}(\lambda)=\Pr[\mathsf{Game}^{\mathsf{adm}}_{\Framework,\mathcal{A}_1}(1^\lambda)=1]$. The framework provides update admissibility if this advantage is negligible for every $\mathsf{PPT}$ adversary $\mathcal{A}_1$.
\end{definition}

\begin{proof}[Proof (Theorem~\ref{thm:update-adm})]
Suppose that $\mathcal{A}_1$ wins the game. Since $\Framework.\mathsf{VfyClient}(\kappa_t^\phi,i,\mu_{i,t})=1$, the checks $\mathsf{Chk}$ succeed and $\Pi.\Verify_{\mathcal{R}_\phi}(\ppnizk,x_{i,t},\pi_{i,t})=1$. Since $\Framework.\mathsf{ValidClient}(\kappa_t^\phi,i,\mu_{i,t})=0$ with the same checks $\mathsf{Chk}$, it follows that $x_{i,t}\notin\mathcal{L}_{\phi}$. An adversary $\mathcal{B}_{\mathsf{adm}}$ that outputs $(x_{i,t},\pi_{i,t})$ therefore breaks soundness of $\Pi$ whenever $\mathcal{A}_1$ wins, which proves the stated bound.
\end{proof}

\subsection{Aggregation Correctness}\label{app:game-agg}
The game in Fig.~\ref{fig:games}(b) formalises Property~\textbf{(P2)}. A corrupted aggregator wins if its aggregation message passes verification but is not valid for the aggregation policy $\psi\in\Phi$ and the recorded admitted set, in a round with active context $\kappa_t^\phi$.

\begin{definition}[Aggregation correctness]
The aggregation-correctness advantage of $\mathcal{A}_2$ is $\operatorname{Adv}^{\mathsf{agg}}_{\Framework,\mathcal{A}_2}(\lambda)=\Pr[\mathsf{Game}^{\mathsf{agg}}_{\Framework,\mathcal{A}_2}(1^\lambda)=1]$. The framework provides aggregation correctness if this advantage is negligible for every $\mathsf{PPT}$ adversary $\mathcal{A}_2$.
\end{definition}

\begin{proof}[Proof (Theorem~\ref{thm:agg-correct})]
Suppose that $\mathcal{A}_2$ wins the game. Since $\Framework.\mathsf{VfyAgg}(\kappa_t^\phi,\psi,\mathsf{Rec}_t,\nu_t)=1$, checks (1) to (4) of Appendix~\ref{app:security-game-notions} succeed. Since check (2) gives distinct identifiers, one of two cases holds.

\noindent\emph{(i)} Some entry $(i,\mu_{i,t})\in\mathsf{Rec}_t$ satisfies $\Framework.\mathsf{ValidClient}(\kappa_t^\phi,i,\mu_{i,t})=0$. By check (3) and the argument for Theorem~\ref{thm:update-adm}, $x_{i,t}\notin\mathcal{L}_{\phi}$ while $\Pi.\Verify_{\mathcal{R}_\phi}(\ppnizk,x_{i,t},\pi_{i,t})=1$. This is a soundness violation of $\Pi$, and a union bound over at most $q(\lambda)$ recorded client messages gives the first term.

\noindent\emph{(ii)} $\Framework.\mathsf{ValidRec}(\kappa_t^\phi,\mathsf{Rec}_t)=1$. Since $\mathcal{A}_2$ still wins, $\Framework.\mathsf{ValidAgg}(\kappa_t^\phi,\psi,\mathsf{Rec}_t,\nu_t)=0$, and therefore $x_t^{\mathsf{agg}}\notin\mathcal{L}_{\psi}$, while $\Pi.\Verify_{\mathcal{R}_\psi}(\ppnizk,x_t^{\mathsf{agg}},\pi_t^{\mathsf{agg}})=1$ by check (4). An adversary that outputs $(x_t^{\mathsf{agg}},\pi_t^{\mathsf{agg}})$ therefore breaks soundness of $\Pi$, which gives the final term.
The guarantee is relative to the public record $\mathsf{Rec}_t$, as stated in Sec.~\ref{sec:threat-model}.
\end{proof}

\subsection{Policy-Relative Privacy}\label{app:game-priv}
The game in Fig.~\ref{fig:games}(c) formalises Property~\textbf{(P3)}. The verifier receives the proof-carrying message for one of two valid private inputs with the same permitted public context and tries to tell which input was used.

\begin{definition}[Policy-relative privacy]
For $\mathcal{A}_2=(\mathcal{A}_{2,1},\mathcal{A}_{2,2})$, the policy-relative-privacy advantage is
\[
\begin{aligned}
\operatorname{Adv}^{\mathsf{priv}}_{\Framework,\mathcal{A}_2}(\lambda)
={}&
\bigl|
\Pr[\mathsf{Game}^{\mathsf{priv}}_{\Framework,\mathcal{A}_2}(1^\lambda,0)=1]
\\
&-
\Pr[\mathsf{Game}^{\mathsf{priv}}_{\Framework,\mathcal{A}_2}(1^\lambda,1)=1]
\bigr|.
\end{aligned}
\]
The framework provides policy-relative privacy if this advantage is negligible for every $\mathsf{PPT}$ verifier $\mathcal{A}_2$.
\end{definition}

\begin{proof}[Proof (Theorem~\ref{thm:policy-privacy})]
Let $(x_0,x_1)$ be an admissible challenge pair with $x_j=(i,t,\phi,\theta_t,C_j,Z^{\mathsf{pub}},\Params)$ for $j\in\{0,1\}$. By definition, $\Framework.\mathsf{CommonCtx}(x_0)=\Framework.\mathsf{CommonCtx}(x_1)$. Consider four hybrids, and write $\Pr[H_i=1]$ for the probability that hybrid $H_i$ outputs $1$.

\noindent\emph{$(H_0)$.} The challenger uses $x_0$ and the real proof $\pi_0\gets\Pi.\Prove_{\mathcal{R}_\phi}(\ppnizk,x_0,w_0)$.

\noindent\emph{$(H_1)$.} The statement remains $x_0$, but the setup and its proof are simulated. By zero knowledge of $\Pi$, $|\Pr[H_0=1]-\Pr[H_1=1]|\leq\operatorname{Adv}^{\mathsf{zk}}_{\Pi,\mathcal{B}_{\mathsf{zk},0}}(\lambda)$.

\noindent\emph{$(H_2)$.} Replace $C_0$ by $C_1$ and keep the setup and proof simulated. The fields in $\Framework.\mathsf{CommonCtx}$ are unchanged, and both commitments use fresh randomness and open to collections structured by the same policy $\phi$. By hiding of $\mathsf{CS}$, $|\Pr[H_1=1]-\Pr[H_2=1]|\leq\operatorname{Adv}^{\mathsf{hid}}_{\mathsf{CS},\mathcal{B}_{\mathsf{hid}}}(\lambda)$.

\noindent\emph{$(H_3)$.} Use the real setup and the real proof $\pi_1\gets\Pi.\Prove_{\mathcal{R}_\phi}(\ppnizk,x_1,w_1)$ for $x_1$. Again by zero knowledge of $\Pi$, $|\Pr[H_2=1]-\Pr[H_3=1]|\leq\operatorname{Adv}^{\mathsf{zk}}_{\Pi,\mathcal{B}_{\mathsf{zk},1}}(\lambda)$.



Hybrids $H_0$ and $H_3$ are the privacy experiment with $b=0$ and $b=1$, so $\operatorname{Adv}^{\mathsf{priv}}_{\Framework,\mathcal{A}_2}(\lambda)=|\Pr[H_0=1]-\Pr[H_3=1]|$, and the triangle inequality gives the claimed bound.
\end{proof}

\section{Detailed-Proofs}
In the following we provide the detailed proofs from aforementioned paragraphs.

\subsection{Cryptographic Assumptions}\label{sec:crypto}
The reductions use soundness and knowledge soundness of $\Pi$, multi-theorem zero knowledge of $\Pi$, and hiding and binding of $\mathsf{CS}$. We use the standard notions~\cite{blum1988nizk,feige1999multiple,bellare1992pok,pedersen1991commitment} and state them for the relation family $\mathcal{R}_{\Phi}$. Every statement of the framework declares its policy. 
The definitions below apply to a single relation $\mathcal{R}\in\mathcal{R}_{\Phi}$. They also apply to the whole family, where the adversary may output a statement for any policy in $\Phi$.

Let $\mathcal{R}$ be such a relation with language $$\mathcal{L}_{\mathcal{R}}=\{x\mid\exists\,w\colon(x,w)\in\mathcal{R}\}$$

\noindent\emph{(i) Completeness.}
For every $(x,w)\in\mathcal{R}$ and $\ppnizk\gets\mathsf{Setup}_{\Pi}(1^\lambda)$, the proof $\pi\gets\Pi.\Prove_{\mathcal{R}}(\ppnizk,x,w)$ satisfies $\Pi.\Verify_{\mathcal{R}}(\ppnizk,x,\pi)=1$, except with negligible probability. Completeness and commitment correctness ensure that honest messages are accepted. 

\noindent\emph{(ii) Soundness and knowledge soundness.}
We use the adaptive versions, in which the adversary chooses the statement after it sees the reference string. Let $\mathcal{B}$ be a $\mathsf{PPT}$ adversary. Let $\mathcal{E}_{\mathcal{B}}$ be a $\mathsf{PPT}$ extractor that receives the same inputs and random coins as $\mathcal{B}$. This non-black-box form of extraction is the usual definition of knowledge soundness for succinct arguments~\cite{bitansky2012snark}.
For $\ppnizk\gets\mathsf{Setup}_{\Pi}(1^\lambda)$, $(x,\pi)\gets\mathcal{B}(\ppnizk)$, $w\gets\mathcal{E}_{\mathcal{B}}(\ppnizk)$, and $v=\Pi.\Verify_{\mathcal{R}}(\ppnizk,x,\pi)$, we define
\[
    \begin{aligned}
        \operatorname{Adv}^{\mathsf{snd}}_{\Pi,\mathcal{B}}(\lambda)
        &=
        \Pr[x\notin\mathcal{L}_{\mathcal{R}}\land v=1],
        \\
        \operatorname{Adv}^{\mathsf{k\text{-}snd}}_{\Pi,\mathcal{B}}(\lambda)
        &=
        \Pr[(x,w)\notin\mathcal{R}\land v=1].
    \end{aligned}
\]
$\Pi$ is sound if $\operatorname{Adv}^{\mathsf{snd}}_{\Pi,\mathcal{B}}(\lambda)$ is negligible for every $\mathsf{PPT}$ adversary $\mathcal{B}$. $\Pi$ is knowledge-sound if every $\mathcal{B}$ has an extractor $\mathcal{E}_{\mathcal{B}}$ for which $\operatorname{Adv}^{\mathsf{k\text{-}snd}}_{\Pi,\mathcal{B}}(\lambda)$ is negligible~\cite{bellare1992pok, bitansky2012snark}.
In Lemma~\ref{lem:committed}, $\mathcal{B}$ also receives the public parameters of the other components, such as $\mathsf{pp}_{\mathsf{CS}}$, and the extractor receives them as well. These parameters are generated honestly and independently of $\ppnizk$. We require knowledge soundness only for this auxiliary input. Extraction for arbitrary common auxiliary input of unbounded length cannot hold under indistinguishability obfuscation for a certain class of circuits~\cite{bitansky2014eowf}, and we do not need it.

Knowledge soundness implies soundness. If $x\notin\mathcal{L}_{\mathcal{R}}$, then $(x,w)\notin\mathcal{R}$ for every $w$. This holds in particular for the output of $\mathcal{E}_{\mathcal{B}}$. The event of the soundness experiment is therefore contained in the event of the knowledge-soundness experiment, and
\[
\operatorname{Adv}^{\mathsf{snd}}_{\Pi,\mathcal{B}}(\lambda)\leq\operatorname{Adv}^{\mathsf{k\text{-}snd}}_{\Pi,\mathcal{B}}(\lambda)
\]
for every extractor $\mathcal{E}_{\mathcal{B}}$. Theorems~\ref{thm:update-adm} and~\ref{thm:agg-correct} use only soundness, because their validity predicates ask for membership $x\in\mathcal{L}_{\phi}$. Lemma~\ref{lem:committed} uses knowledge soundness to show that a client policy holds on the committed objects that the prover actually knows.

\noindent\emph{(iii) Multi-theorem zero knowledge.}
Let $(\mathsf{SimSetup},\mathsf{SimProve})$ be the simulator of $\Pi$.
In $\mathsf{ZKReal}_{\Pi,\mathcal{B}}(1^\lambda)$, the challenger computes $\ppnizk\gets\mathsf{Setup}_{\Pi}(1^\lambda)$ and runs $\mathcal{B}(\ppnizk)$.
Whenever $\mathcal{B}$ submits a pair $(x,w)$, the challenger returns $\bot$ if $(x,w)\notin\mathcal{R}$. It returns $\pi\gets\Pi.\Prove_{\mathcal{R}}(\ppnizk,x,w)$ otherwise.
The experiment outputs the bit that $\mathcal{B}$ returns. The experiment $\mathsf{ZKSim}_{\Pi,\mathcal{B}}(1^\lambda)$ is identical, except for two changes. The challenger computes $(\widetilde{\mathsf{pp}}_{\Pi},\tau)\gets\mathsf{SimSetup}(1^\lambda)$ and gives $\widetilde{\mathsf{pp}}_{\Pi}$ to $\mathcal{B}$. It answers every pair $(x,w)\in\mathcal{R}$ with $\widetilde{\pi}\gets\mathsf{SimProve}_{\mathcal{R}}(\tau,x)$.
The adversary may submit polynomially many adaptively chosen pairs under one reference string. This is adaptive multi-theorem zero knowledge in the common reference string model~\cite{feige1999multiple}.
The zero-knowledge advantage is
\[
    \begin{aligned}
        \operatorname{Adv}^{\mathsf{zk}}_{\Pi,\mathcal{B}}(\lambda)
        ={}&
        \bigl|
        \Pr[\mathsf{ZKReal}_{\Pi,\mathcal{B}}(1^\lambda)=1]
        \\
        &-
        \Pr[\mathsf{ZKSim}_{\Pi,\mathcal{B}}(1^\lambda)=1]
        \bigr|,
    \end{aligned}
\]
and $\Pi$ is zero knowledge if this advantage is negligible for every $\mathsf{PPT}$ adversary $\mathcal{B}$.
A single query gives single-theorem zero knowledge, which suffices for Theorem~\ref{thm:policy-privacy}. The multi-theorem form is needed when one reference string serves many proofs, as in Corollary~\ref{cor:multi-priv}. The public parameters are generated honestly, and the setup erases trapdoor information. No participant therefore holds this trapdoor $\tau$, and only the reductions use it.

\noindent\emph{(iv) Commitment correctness, hiding, and binding.}
We write $c=\Commit(\mathsf{pp}_{\mathsf{CS}},M;R)$ when $R$ opens $c$ to $M$. Correctness requires that $(C,R)\gets\Commit(\mathsf{pp}_{\mathsf{CS}},M)$ implies $C=\Commit(\mathsf{pp}_{\mathsf{CS}},M;R)$.

In the hiding experiment, a two-stage $\mathsf{PPT}$ adversary $\mathcal{B}=(\mathcal{B}_1,\mathcal{B}_2)$ plays against a challenger. The challenger computes $\mathsf{pp}_{\mathsf{CS}}\gets\mathsf{Setup}_{\mathsf{CS}}(1^\lambda)$, and $\mathcal{B}_1(\mathsf{pp}_{\mathsf{CS}})$ outputs two messages $M_0,M_1$ of the same format and a state $\mathsf{st}$. For collections, the same format means the same length and the same entry domains.
The challenger samples $b\gets\{0,1\}$ and computes $(C_b,R_b)\gets\Commit(\mathsf{pp}_{\mathsf{CS}},M_b)$. Then $\mathcal{B}_2(\mathsf{st},C_b)$ outputs $b'$.
The hiding advantage is $$\operatorname{Adv}^{\mathsf{hid}}_{\mathsf{CS},\mathcal{B}}(\lambda)=\left|2\Pr[b'=b]-1\right|$$
The binding advantage $\operatorname{Adv}^{\mathsf{bind}}_{\mathsf{CS},\mathcal{B}}(\lambda)$ is the probability that $\mathcal{B}(\mathsf{pp}_{\mathsf{CS}})$ outputs $(M,R,M',R')$ with $M\neq M'$ and $\Commit(\mathsf{pp}_{\mathsf{CS}},M;R)=\Commit(\mathsf{pp}_{\mathsf{CS}},M';R')$.

A collection is committed entrywise. Hiding and binding for collections therefore follow from those for a single entry. A hybrid over the $J$ entries bounds the hiding advantage for collections by $J$ times the single-entry advantage. Two different collections differ in some entry, which gives a binding break for that entry. We state the advantages directly for collections.

\noindent\emph{(v) Auxiliary mechanisms ($\Aux$).}
Auxiliary outputs enter the statements as public values $Z^{\mathsf{pub}}$, and the policy relations encode their consistency. The proofs below therefore need no assumption on $\Aux$. A guarantee that an auxiliary primitive adds, such as the confidentiality of encrypted updates, relies on the assumptions of the instantiation (Def.~5 of the paper). When a validity statement refers to the result of an auxiliary computation, such as a threshold decryption, it uses the correctness of that mechanism.

\subsection{Security Game Notions and Notation}\label{sec:notation}
In order to make the rest of section intuitive we explain few notions. The games separate \emph{verification} from \emph{semantic validity}. Verification methods $\Framework.\mathsf{Vfy}(\cdot)$ are the efficient checks that honest parties perform during the protocol. Validity predicates $\Framework.\mathsf{Valid}(\cdot)$ state the policy conditions that an accepted message is meant to satisfy. They are used only in the definitions and need not be efficiently decidable. An adversary wins a soundness game when verification accepts a message that is not valid.

\noindent\textbf{Simulating the setup.}
In every game, the challenger runs $\Framework.\mathsf{Setup}$ honestly. The setup initialises $\mathsf{CS}$, $\Pi$, and each auxiliary mechanism independently. A reduction that receives $\ppnizk$ or $\mathsf{pp}_{\mathsf{CS}}$ from its own challenger can therefore generate the remaining setup outputs itself. It obtains the same distribution of $\mathsf{pp}_{\Framework}$ and of the setup states. The active context of a round is public and fixed by the honest execution. The reduction can therefore compute it in the same way as the challenger.

\noindent\textbf{Client messages.}
Fix a round $t$ and a client policy $\phi$, and let $\kappa_t^\phi=(t,\phi,\theta_t,\Params)$ denote the active public context.
A client message $\mu_{i,t}=(\hat{i},\hat{t},x_{i,t},\pi_{i,t})$ is received from client $i$ over an authenticated channel.
The channel therefore determines identity $i$, not the message. Client verification and client validity share the deterministic checks
\[
\begin{aligned}
\mathsf{Chk}(\kappa_t^\phi,i,\mu_{i,t})=1
\Longleftrightarrow{}&
\hat{i}=i
\land
\hat{t}=t
\\
&\land
\mathsf{Ctx}(x_{i,t})=(i,t,\phi,\theta_t,\Params)
\\
&\land
\mathsf{ExtractAgg}_{\mathsf{Type},\phi}(x_{i,t})\neq\bot .
\end{aligned}
\]
They differ only in the final condition.
\[
\begin{aligned}
&\Framework.\mathsf{VfyClient}(\kappa_t^\phi,i,\mu_{i,t})=1
\Longleftrightarrow
\mathsf{Chk}(\kappa_t^\phi,i,\mu_{i,t})=1
\\
&\qquad\qquad\land
\Pi.\Verify_{\mathcal{R}_\phi}(\ppnizk,x_{i,t},\pi_{i,t})=1,
\\[1mm]
&\Framework.\mathsf{ValidClient}(\kappa_t^\phi,i,\mu_{i,t})=1
\Longleftrightarrow
\mathsf{Chk}(\kappa_t^\phi,i,\mu_{i,t})=1
\\
&\qquad\qquad\land
x_{i,t}\in\mathcal{L}_{\phi}.
\end{aligned}
\]
Fig.~\ref{fig:client-vfy-valid} writes out both predicates side by side. The first four conditions in each column are the checks $\mathsf{Chk}$, and the boxed conditions are the only difference.
Hence, if $\mathsf{Chk}(\kappa_t^\phi,i,\mu_{i,t})=1$, then $\Framework.\mathsf{ValidClient}(\kappa_t^\phi,i,\mu_{i,t})=0$ holds exactly when $x_{i,t}\notin\mathcal{L}_{\phi}$. We use this in all soundness proofs.

\begin{figure}[!t]
\centering
\footnotesize
\setlength{\tabcolsep}{3pt}
\renewcommand{\arraystretch}{1.08}
\setlength{\fboxsep}{2pt}%
\fbox{%
\begin{minipage}{\dimexpr\columnwidth-2\fboxsep-2\fboxrule\relax}
\begin{tabularx}{\linewidth}{
    @{}>{\arraybackslash}X
    @{\hspace{1.5pt}\vrule\hspace{1.5pt}}
    >{\arraybackslash}X@{}
}
\textbf{$\Framework.\mathsf{VfyClient}$}
&
\textbf{$\Framework.\mathsf{ValidClient}$}
\\[2pt]
\hline
\\[-5pt]
\(
\begin{aligned}
&\Framework.\mathsf{VfyClient}(\kappa_t^\phi,i,\mu_{i,t})=1
\\[1mm]
&\Longleftrightarrow
\\
&\hat{i}=i {}\land{}
\\
&\hat{t}=t {}\land{}
\\
&\mathsf{Ctx}(x_{i,t})=(i,t,\phi,\theta_t,\Params) {}\land{}
\\
&\mathsf{ExtractAgg}_{\mathsf{Type},\phi}(x_{i,t})\neq\bot {}\land{}
\\
&\boxed{\Pi.\Verify_{\mathcal{R}_\phi}(\ppnizk,x_{i,t},\pi_{i,t})=1}
\end{aligned}
\)
&
\(
\begin{aligned}
&\Framework.\mathsf{ValidClient}(\kappa_t^\phi,i,\mu_{i,t})=1
\\[1mm]
&\Longleftrightarrow
\\
&\hat{i}=i {}\land{}
\\
&\hat{t}=t {}\land{}
\\
&\mathsf{Ctx}(x_{i,t})=(i,t,\phi,\theta_t,\Params) {}\land{}
\\
&\mathsf{ExtractAgg}_{\mathsf{Type},\phi}(x_{i,t})\neq\bot {}\land{}
\\
&\boxed{x_{i,t}\in\mathcal{L}_{\phi}}
\end{aligned}
\)
\end{tabularx}
\end{minipage}%
}
\caption{Client verification and client validity for a fixed round. The first four conditions are the shared checks $\mathsf{Chk}$. Verification ends with the efficient proof check, whereas validity ends with membership in $\mathcal{L}_{\phi}$.}
\label{fig:client-vfy-valid}
\end{figure}

\noindent\textbf{Records and aggregation messages.}
Let $\mathsf{Rec}_t=\bigl((i_1,\mu_{i_1,t}),\ldots,(i_m,\mu_{i_m,t})\bigr)$ denote the public admitted-message record $\mathsf{Adm}_t$ of round $t$ (Algorithm~3 of the paper), with each message paired with its authenticated sender identity $i_j$.
$\Framework.\mathsf{ValidRec}(\kappa_t^\phi,\mathsf{Rec}_t)=1$ if and only if the identifiers in $\mathsf{Rec}_t$ are pairwise distinct and $\Framework.\mathsf{ValidClient}(\kappa_t^\phi,i,\mu_{i,t})=1$ for every $(i,\mu_{i,t})\in\mathsf{Rec}_t$.
We write $\Framework.\mathsf{AggOut}_{\mathsf{Type},\psi}(\theta_t,\mathsf{Rec}_t,\Params)$ for the deterministic output that extracts exactly one aggregation value from every recorded message, applies the aggregation rule and weights prescribed by $\mathsf{Type}$ and $\psi$, and applies the prescribed model-update rule.
An aggregation message has the form $\nu_t=(\ell,\hat{t},x_t^{\mathsf{agg}},\pi_t^{\mathsf{agg}})$
with
\[
x_t^{\mathsf{agg}}=(\hat{\ell},\tilde{t},\hat{\phi},\hat{\psi},\hat{\theta},\widehat{\mathsf{Rec}}_t,\theta_{t+1},Z_t^{\mathsf{pub}},\widehat{\Params}).
\]
We write $\Framework.\mathsf{VfyAgg}(\kappa_t^\phi,\psi,\mathsf{Rec}_t,\nu_t)=1$ if and only if four checks succeed.
\begin{enumerate}[label=(\arabic*),leftmargin=*,labelindent=0pt]
\item The message and the statement match the active context, namely $\hat{t}=\tilde{t}=t$, $\hat{\ell}=\ell$, $\hat{\phi}=\phi$, $\hat{\psi}=\psi$, $\hat{\theta}=\theta_t$, $\widehat{\Params}=\Params$, and $\widehat{\mathsf{Rec}}_t=\mathsf{Rec}_t$.
\item The identifiers in $\mathsf{Rec}_t$ are pairwise distinct.
\item $\Framework.\mathsf{VfyClient}(\kappa_t^\phi,i,\mu_{i,t})=1$ for every $(i,\mu_{i,t})\in\mathsf{Rec}_t$.
\item $\Pi.\Verify_{\mathcal{R}_{\psi}}(\ppnizk,x_t^{\mathsf{agg}},\pi_t^{\mathsf{agg}})=1$.
\end{enumerate}

The corresponding validity predicate $$\Framework.\mathsf{ValidAgg}(\kappa_t^\phi,\psi,\mathsf{Rec}_t,\nu_t)=1$$ holds if and only if 
$$\Framework.\mathsf{ValidRec}(\kappa_t^\phi,\mathsf{Rec}_t)=1 \quad\text{and}\quad x_t^{\mathsf{agg}}\in\mathcal{L}_{\psi}$$

\noindent\textbf{Consequences of validity.}
The relation $\mathcal{R}_{\psi}$ enforces the prescribed participation, extraction, weighting, aggregation, and model update. Hence $x_t^{\mathsf{agg}}\in\mathcal{L}_{\psi}$ implies
\[
\theta_{t+1}=\Framework.\mathsf{AggOut}_{\mathsf{Type},\psi}(\theta_t,\mathsf{Rec}_t,\Params).
\]
For a realisation with an auxiliary mechanism, this implication also uses the correctness of that mechanism, such as correct homomorphic evaluation and threshold decryption.
If in addition $\Framework.\mathsf{ValidRec}(\kappa_t^\phi,\mathsf{Rec}_t)=1$, then every recorded message is valid and every identifier appears once. The published model is then obtained from every recorded contribution exactly once, with the prescribed rule and weights. This excludes omission, duplication, substitution, modification, and reweighting within the record, as stated by Property~(P2).

\noindent\textbf{Privacy notation.}
Let $s$ denote the private client input and local state required by the selected FL type, client policy, and realisation. We write $(x,w)\gets\Framework.\mathsf{BuildClient}(\kappa_t^\phi,i,s)$ for running Algorithm~2 of the paper up to the formation of the statement and the witness, before the relation check and the proof generation.
For a client statement
\[
x_{i,t}=(i,t,\phi,\theta_t,C_{i,t},Z_{i,t}^{\mathsf{pub}},\Params),
\]
we define
\[
\Framework.\mathsf{CommonCtx}(x_{i,t})=(i,t,\phi,\theta_t,Z_{i,t}^{\mathsf{pub}},\Params).
\]
These are the public fields that must be identical in the two privacy worlds. The commitment $C_{i,t}$ is excluded from them but remains part of the statement.
$\Framework.\mathsf{BuildClient}$ derives from $s$ the collection $M$ to be committed, the auxiliary values, and the remaining witness values. It commits to $M$ with fresh randomness and places $C$ in the statement and $R$ in the witness. The commitment randomness is independent of all other values. In particular, it does not affect whether a pair of inputs is admissible in the privacy game. We call a challenge pair \emph{admissible} if it passes the check of the privacy game in Fig.~\ref{fig:game-priv}. This term is unrelated to update admissibility (P1). $\Framework.\mathsf{CommonCtx}$ excludes $C$, and an honestly generated commitment satisfies the commitment conditions of $\mathcal{R}_{\phi}$ by correctness.

\subsection{Update Admissibility}\label{sec:adm}
We formalise Property~(P1). The game captures a corrupted client whose message passes client verification although it is not valid under the active client policy. Fix a round $t$ of an honest execution, and let $\kappa_t^\phi$ be its active public context.

\begin{figure}[t]
\centering
\small
\GameBox{Game
$\mathsf{Game}^{\mathsf{adm}}_{\Framework,\mathcal{A}_1}
(1^\lambda)$}{%
$\triangleright$ \textbf{Setup}\\
Challenger $\mathcal{C}$ computes
\(
(\mathsf{pp}_{\Framework},
 \{\mathsf{st}^{\mathsf{setup}}_P\}_{P\in\mathcal P})
\gets
\Framework.\mathsf{Setup}(1^\lambda).
\)
$\mathcal{C}$ gives
$(\mathsf{pp}_{\Framework},\kappa_t^\phi)$
and the prescribed setup state of the corrupted client to $\mathcal{A}_1$.

$\triangleright$ \textbf{Adversary output}\\
The adversary, acting as corrupted client $i$, outputs
\(
\mu_{i,t}
\gets
\mathcal{A}_1
(\mathsf{pp}_{\Framework},\kappa_t^\phi).
\)

\textbf{Return} $1$ if and only if
\[
\Framework.\mathsf{VfyClient}
(\kappa_t^\phi,i,\mu_{i,t})=1
\,\land\,
\Framework.\mathsf{ValidClient}
(\kappa_t^\phi,i,\mu_{i,t})=0.
\]
}
\caption{Update-admissibility game.}
\label{fig:game-adm}
\end{figure}

\setcounter{definition}{7}
\begin{definition}[Update admissibility]\label{def:adm}
The update-admissibility advantage of $\mathcal{A}_1$ is
\[
\operatorname{Adv}^{\mathsf{adm}}_{\Framework,\mathcal{A}_1}(\lambda)
=
\Pr[
\mathsf{Game}^{\mathsf{adm}}_{\Framework,\mathcal{A}_1}
(1^\lambda)=1
].
\]
The framework provides update admissibility if this advantage is negligible for every $\mathsf{PPT}$ adversary $\mathcal{A}_1$.
\end{definition}

\begin{theorem}[Update admissibility]\label{thm:update-adm}
For every $\mathsf{PPT}$ adversary $\mathcal{A}_1$, there exists a $\mathsf{PPT}$ adversary $\mathcal{B}_{\mathsf{adm}}$ such that
\[
\operatorname{Adv}^{\mathsf{adm}}_{\Framework,\mathcal{A}_1}(\lambda)
\leq
\operatorname{Adv}^{\mathsf{snd}}_{\Pi,\mathcal{B}_{\mathsf{adm}}}(\lambda).
\]
Consequently, the framework provides update admissibility whenever $\Pi$ is sound.
\end{theorem}

\begin{proof}
We build an adversary $\mathcal{B}_{\mathsf{adm}}$ against the soundness of $\Pi$ for the relation $\mathcal{R}_{\phi}$.
$\mathcal{B}_{\mathsf{adm}}$ receives $\ppnizk$ from its challenger. It generates the remaining setup outputs and computes $\kappa_t^\phi$ as described in Section~\ref{sec:notation}. $\mathcal{A}_1$ therefore receives exactly the inputs of the game. $\mathcal{B}_{\mathsf{adm}}$ runs $\mathcal{A}_1$ with the prescribed setup state of the corrupted client and obtains $\mu_{i,t}=(\hat{i},\hat{t},x_{i,t},\pi_{i,t})$. It outputs $(x_{i,t},\pi_{i,t})$. $\mathcal{B}_{\mathsf{adm}}$ runs the setup and $\mathcal{A}_1$ once, so it runs in polynomial time.

Suppose that $\mathcal{A}_1$ wins the game. Then
\[
\Framework.\mathsf{VfyClient}
(\kappa_t^\phi,i,\mu_{i,t})=1 .
\]
The identity, round, context, and extraction checks therefore succeed, that is, $\mathsf{Chk}(\kappa_t^\phi,i,\mu_{i,t})=1$. The proof also verifies, so
\(
\Pi.\Verify_{\mathcal{R}_\phi}(\ppnizk,x_{i,t},\pi_{i,t})=1.
\)
The adversary also has
\(
\Framework.\mathsf{ValidClient}
(\kappa_t^\phi,i,\mu_{i,t})=0 .
\)
Client validity uses the same checks $\mathsf{Chk}$, and these checks succeed (Fig.~\ref{fig:client-vfy-valid}). Hence
\(
x_{i,t}\notin\mathcal{L}_{\phi}.
\)
Together,
\[
x_{i,t}\notin\mathcal{L}_{\phi}
\quad\text{and}\quad
\Pi.\Verify_{\mathcal{R}_\phi}(\ppnizk,x_{i,t},\pi_{i,t})=1.
\]
The output of $\mathcal{B}_{\mathsf{adm}}$ therefore breaks the soundness of $\Pi$ whenever $\mathcal{A}_1$ wins. This gives
\[
\operatorname{Adv}^{\mathsf{adm}}_{\Framework,\mathcal{A}_1}(\lambda)
\leq
\operatorname{Adv}^{\mathsf{snd}}_{\Pi,\mathcal{B}_{\mathsf{adm}}}(\lambda).
\qedhere
\]
\end{proof}

\subsection{Aggregation Correctness}\label{sec:agg}
We formalise Property~(P2). The game captures a corrupted aggregator whose aggregation message passes verification although it is not valid for the aggregation policy and the recorded admitted set. Fix the active client context $\kappa_t^\phi$ and the aggregation policy $\psi\in\Phi$ of the round.

\begin{figure}[t]
\centering
\small
\GameBox{Game
$\mathsf{Game}^{\mathsf{agg}}_{\Framework,\mathcal{A}_2}
(1^\lambda)$}{%
$\triangleright$ \textbf{Setup}\\
Challenger $\mathcal{C}$ computes
\(
(\mathsf{pp}_{\Framework},
 \{\mathsf{st}^{\mathsf{setup}}_P\}_{P\in\mathcal P})
\gets
\Framework.\mathsf{Setup}(1^\lambda).
\)
$\mathcal{C}$ gives
\(
(\mathsf{pp}_{\Framework},\kappa_t^\phi,\psi)
\)
and the prescribed setup state of the corrupted aggregator to
$\mathcal{A}_2$.

$\triangleright$ \textbf{Adversary output}\\
The adversarial aggregator outputs
\[
(\mathsf{Rec}_t,\nu_t)
\gets
\mathcal{A}_2
(\mathsf{pp}_{\Framework},\kappa_t^\phi,\psi).
\]
\textbf{Return} $1$ if and only if
\[
\begin{aligned}
&\Framework.\mathsf{VfyAgg}
(\kappa_t^\phi,\psi,\mathsf{Rec}_t,\nu_t)=1
\\
&\quad {}\land
\Framework.\mathsf{ValidAgg}
(\kappa_t^\phi,\psi,\mathsf{Rec}_t,\nu_t)=0.
\end{aligned}
\]
}
\caption{Aggregation-correctness game.}
\label{fig:game-agg}
\end{figure}

\begin{definition}[Aggregation correctness]\label{def:agg}
The aggregation-correctness advantage of $\mathcal{A}_2$ is
\[
\operatorname{Adv}^{\mathsf{agg}}_{\Framework,\mathcal{A}_2}(\lambda)
=
\Pr[
\mathsf{Game}^{\mathsf{agg}}_{\Framework,\mathcal{A}_2}
(1^\lambda)=1
].
\]
The framework provides aggregation correctness if this advantage is negligible for every $\mathsf{PPT}$ adversary $\mathcal{A}_2$.
\end{definition}

\begin{theorem}[Aggregation correctness]\label{thm:agg-correct}
Let at most $q(\lambda)$ client messages appear in $\mathsf{Rec}_t$.
For every $\mathsf{PPT}$ adversary $\mathcal{A}_2$, there exist $\mathsf{PPT}$ adversaries
$\mathcal{B}_{1},\ldots,\mathcal{B}_{q}$ and $\mathcal{B}_{\mathsf{agg}}$ such that
\[
\operatorname{Adv}^{\mathsf{agg}}_{\Framework,\mathcal{A}_2}(\lambda)
\leq
\sum_{j=1}^{q(\lambda)}
\operatorname{Adv}^{\mathsf{snd}}_{\Pi,\mathcal{B}_{j}}(\lambda)
+
\operatorname{Adv}^{\mathsf{snd}}_{\Pi,\mathcal{B}_{\mathsf{agg}}}(\lambda).
\]
Hence the aggregation-correctness advantage is negligible when $q$ is polynomial and $\Pi$ is sound.
\end{theorem}

\begin{proof}
Let $W$ be the event that $\mathcal{A}_2$ wins. Write $\mathsf{Rec}_t=\bigl((i_1,\mu_{i_1,t}),\ldots,(i_m,\mu_{i_m,t})\bigr)$ with $m\leq q(\lambda)$.
If $W$ occurs, then
\(
\Framework.\mathsf{VfyAgg}
(\kappa_t^\phi,\psi,\mathsf{Rec}_t,\nu_t)=1 .
\)
Checks (1) to (4) of Section~\ref{sec:notation} therefore succeed. The record contains distinct identifiers, every recorded message passes $\Framework.\mathsf{VfyClient}$, and the aggregation proof is accepted.

We split $W$ into events. For $j\in\{1,\ldots,q(\lambda)\}$, let $E_j$ be the event that $W$ occurs, the record has at least $j$ entries, the $j$-th recorded message is invalid, and all earlier recorded messages are valid. Formally, $\Framework.\mathsf{ValidClient}(\kappa_t^\phi,i_j,\mu_{i_j,t})=0$ and $\Framework.\mathsf{ValidClient}(\kappa_t^\phi,i_k,\mu_{i_k,t})=1$ for all $k<j$. Let $E_{\mathsf{agg}}$ be the event that $W$ occurs and $\Framework.\mathsf{ValidRec}(\kappa_t^\phi,\mathsf{Rec}_t)=1$.
These events cover $W$. Suppose that $W$ occurs and $\Framework.\mathsf{ValidRec}(\kappa_t^\phi,\mathsf{Rec}_t)=0$. Check (2) rules out repeated identifiers. Some recorded message is therefore invalid, and $E_j$ holds for the smallest such index $j$. Hence
\[
\Pr[W]
\leq
\sum_{j=1}^{q(\lambda)}\Pr[E_j]+\Pr[E_{\mathsf{agg}}].
\]

\noindent\emph{(i) Invalid recorded client message.}
For each $j$, the adversary $\mathcal{B}_j$ receives $\ppnizk$ from its soundness challenger for $\mathcal{R}_{\phi}$. It simulates the setup and the context as in Section~\ref{sec:notation} and runs $\mathcal{A}_2$. It outputs $(x_{i_j,t},\pi_{i_j,t})$ from the $j$-th recorded message, or $\bot$ if the record has fewer than $j$ entries.
Suppose that $E_j$ occurs. Check (3) gives
\(
\Framework.\mathsf{VfyClient}
(\kappa_t^\phi,i_j,\mu_{i_j,t})=1 .
\)
Hence $\mathsf{Chk}(\kappa_t^\phi,i_j,\mu_{i_j,t})=1$ and
\(
\Pi.\Verify_{\mathcal{R}_\phi}(\ppnizk,x_{i_j,t},\pi_{i_j,t})=1.
\)
The client validity of this message is $0$ while $\mathsf{Chk}$ succeeds. Hence
\(
x_{i_j,t}\notin\mathcal{L}_{\phi}.
\)
This is breaks the soundness of $\Pi$ and therefore,
\[
\Pr[E_j]\leq\operatorname{Adv}^{\mathsf{snd}}_{\Pi,\mathcal{B}_j}(\lambda).
\]
Summing over at most $q(\lambda)$ recorded client messages gives the first term.

\noindent\emph{(ii) Valid record.}
The adversary $\mathcal{B}_{\mathsf{agg}}$ receives $\ppnizk$ from its soundness challenger for $\mathcal{R}_{\psi}$. It simulates the setup and the context in the same way, runs $\mathcal{A}_2$, and outputs $(x_t^{\mathsf{agg}},\pi_t^{\mathsf{agg}})$.
Suppose that $E_{\mathsf{agg}}$ occurs. Then
\[
\begin{aligned}
&\Framework.\mathsf{ValidRec}
(\kappa_t^\phi,\mathsf{Rec}_t)=1
\quad\text{and}
\\
&\Framework.\mathsf{ValidAgg}
(\kappa_t^\phi,\psi,\mathsf{Rec}_t,\nu_t)=0.
\end{aligned}
\]
By the definition of $\Framework.\mathsf{ValidAgg}$, this gives
\(
x_t^{\mathsf{agg}}\notin\mathcal{L}_{\psi}.
\)
At the same time, check (4) gives
\(
\Pi.\Verify_{\mathcal{R}_\psi}
(\ppnizk,x_t^{\mathsf{agg}},\pi_t^{\mathsf{agg}})=1.
\)
The output of $\mathcal{B}_{\mathsf{agg}}$ therefore breaks the soundness of $\Pi$, and
\[
\Pr[E_{\mathsf{agg}}]\leq\operatorname{Adv}^{\mathsf{snd}}_{\Pi,\mathcal{B}_{\mathsf{agg}}}(\lambda).
\]
This gives the final term.

All adversaries run $\mathcal{A}_2$ once and are therefore $\mathsf{PPT}$. Combining the two cases gives
\[
\operatorname{Adv}^{\mathsf{agg}}_{\Framework,\mathcal{A}_2}(\lambda)
\leq
\sum_{j=1}^{q(\lambda)}
\operatorname{Adv}^{\mathsf{snd}}_{\Pi,\mathcal{B}_{j}}(\lambda)
+
\operatorname{Adv}^{\mathsf{snd}}_{\Pi,\mathcal{B}_{\mathsf{agg}}}(\lambda).
\]
When the aggregation message is valid, the consequences of validity in Section~\ref{sec:notation} apply. The guarantee is relative to the public record $\mathsf{Rec}_t$. Whether every submitted message entered this record is outside this game. An auditable submission mechanism can provide that guarantee.
\end{proof}

\subsection{Policy-Relative Privacy}\label{sec:priv}
We formalise Property~(P3). The game compares two valid private client inputs that induce the same permitted public context. The verifier receives the proof-carrying message for one of them and tries to tell which input was used.

\begin{figure}[t]
\centering
\small
\GameBox{Game
$\mathsf{Game}^{\mathsf{priv}}_{\Framework,\mathcal{A}_2}
(1^\lambda,b)$}{%
$\triangleright$ \textbf{Setup}\\
Challenger $\mathcal{C}$ computes
\(
(\mathsf{pp}_{\Framework},
 \{\mathsf{st}^{\mathsf{setup}}_P\}_{P\in\mathcal P})
\gets
\Framework.\mathsf{Setup}(1^\lambda).
\)
$\mathcal{C}$ gives
\(
(\mathsf{pp}_{\Framework},\kappa_t^\phi)
\)
and the prescribed setup state of the corrupted participant to $\mathcal{A}_{2,1}$.

\medskip
$\triangleright$ \textbf{Challenge selection}\\
The first-stage adversary chooses a client $i$ and two private inputs $s_0,s_1$, each holding the client data and local state required by $\phi$. It outputs
\(
(i,s_0,s_1,\mathsf{st})
\gets
\mathcal{A}_{2,1}
(\mathsf{pp}_{\Framework},\kappa_t^\phi).
\)

\medskip
$\triangleright$ \textbf{Challenge construction}\\
For each $j\in\{0,1\}$, $\mathcal{C}$ computes
\(
(x_j,w_j)
\gets
\Framework.\mathsf{BuildClient}
(\kappa_t^\phi,i,s_j).
\)

\medskip
$\triangleright$ \textbf{Admissibility}\\
If
\[
\Framework.\mathsf{CommonCtx}(x_0)
\neq
\Framework.\mathsf{CommonCtx}(x_1)
\]
or
\[
(x_j,w_j)\notin\mathcal{R}_{\phi}
\quad
\text{for some }j\in\{0,1\},
\]
\textbf{Return} $0$.

\medskip
$\triangleright$ \textbf{Challenge}\\
$\mathcal{C}$ computes
\(
\pi_b
\gets
\Pi.\Prove_{\mathcal{R}_\phi}(\ppnizk,x_b,w_b)
\)
and gives
\(
(x_b,\pi_b)
\)
to $\mathcal{A}_{2,2}$.

\medskip
$\triangleright$ \textbf{Adversary output}\\
Besides its state $\mathsf{st}$, the second-stage adversary sees only the statement and the proof. It outputs
\[
b'
\gets
\mathcal{A}_{2,2}
(\mathsf{st},x_b,\pi_b).
\]
\textbf{Return} $b'$.
}
\caption{Policy-relative-privacy game.}
\label{fig:game-priv}
\end{figure}

\begin{definition}[Policy-relative privacy]\label{def:priv}
Let $\mathcal{A}_2=(\mathcal{A}_{2,1},\mathcal{A}_{2,2})$. The policy-relative-privacy advantage is
\[
\begin{aligned}
\operatorname{Adv}^{\mathsf{priv}}_{\Framework,\mathcal{A}_2}(\lambda)
=
\Bigl|
&
\Pr[
\mathsf{Game}^{\mathsf{priv}}_{\Framework,\mathcal{A}_2}
(1^\lambda,0)=1
]
\\
-
&
\Pr[
\mathsf{Game}^{\mathsf{priv}}_{\Framework,\mathcal{A}_2}
(1^\lambda,1)=1
]
\Bigr|.
\end{aligned}
\]
The framework provides policy-relative privacy if this advantage is negligible for every $\mathsf{PPT}$ verifier $\mathcal{A}_2$.
\end{definition}

\begin{theorem}[Policy-relative privacy]\label{thm:policy-privacy}
For every $\mathsf{PPT}$ verifier $\mathcal{A}_2$, there exist
$\mathsf{PPT}$ adversaries
$\mathcal{B}_{\mathsf{zk},0}$,
$\mathcal{B}_{\mathsf{hid}}$, and
$\mathcal{B}_{\mathsf{zk},1}$ such that
\[
\begin{aligned}
\operatorname{Adv}^{\mathsf{priv}}_{\Framework,\mathcal{A}_2}(\lambda)
\leq{}&
\operatorname{Adv}^{\mathsf{zk}}_{\Pi,\mathcal{B}_{\mathsf{zk},0}}(\lambda)
+
\operatorname{Adv}^{\mathsf{hid}}_{\mathsf{CS},\mathcal{B}_{\mathsf{hid}}}
(\lambda)
\\
&+
\operatorname{Adv}^{\mathsf{zk}}_{\Pi,\mathcal{B}_{\mathsf{zk},1}}(\lambda).
\end{aligned}
\]
Consequently, the framework provides policy-relative privacy when $\Pi$ is zero knowledge and $\mathsf{CS}$ is hiding.
\end{theorem}

\begin{proof}
Let $(x_0, x_1)$ be an admissible challenge pair with
\[
\begin{aligned}
x_0&=(i,t,\phi,\theta_t,C_0,Z^{\mathsf{pub}},\Params),\\
x_1&=(i,t,\phi,\theta_t,C_1,Z^{\mathsf{pub}},\Params).
\end{aligned}
\]
By admissibility,
\[
\Framework.\mathsf{CommonCtx}(x_0)
=
\Framework.\mathsf{CommonCtx}(x_1).
\]
The two statements therefore differ only in their commitments.

We use four hybrids. Each hybrid runs the game up to the admissibility check exactly as the game does. It returns $0$ if the pair is not admissible. The hybrids differ only in the setup of $\Pi$, in the statement given to $\mathcal{A}_{2,2}$, and in how its proof is produced. We write $\Pr[H_i=1]$ for the probability that hybrid $H_i$ outputs $1$.

\noindent\emph{$(H_0)$ Real proof for $x_0$.}
The challenger uses $\ppnizk\gets\mathsf{Setup}_{\Pi}(1^\lambda)$, the statement $x_0$, and the real proof
\[
\pi_0
\gets
\Pi.\Prove_{\mathcal{R}_\phi}(\ppnizk,x_0,w_0).
\]
This is the privacy game with $b=0$.

\noindent\emph{$(H_1)$ Simulated proof for $x_0$.}
The challenger uses $(\widetilde{\mathsf{pp}}_{\Pi},\tau)\gets\mathsf{SimSetup}(1^\lambda)$. The statement remains $x_0$, but its proof is the simulated proof $\widetilde{\pi}\gets\mathsf{SimProve}_{\mathcal{R}_\phi}(\tau,x_0)$.

The adversary $\mathcal{B}_{\mathsf{zk},0}$ moves from $H_0$ to $H_1$. It receives a reference string from its zero-knowledge challenger and generates the remaining setup outputs itself. It runs $\mathcal{A}_{2,1}$ and computes $(x_j,w_j)\gets\Framework.\mathsf{BuildClient}(\kappa_t^\phi,i,s_j)$ for $j\in\{0,1\}$. If the pair is not admissible, it outputs $0$. Else, it submits $(x_0,w_0)$, which lies in $\mathcal{R}_{\phi}$ by admissibility. It receives a proof, runs $\mathcal{A}_{2,2}$ on $x_0$ and this proof, and outputs the bit of $\mathcal{A}_{2,2}$. In $\mathsf{ZKReal}$, $\mathcal{B}_{\mathsf{zk},0}$ runs $H_0$. In $\mathsf{ZKSim}$, it runs $H_1$. By zero knowledge of $\Pi$,
\[
\left|
\Pr[H_0=1]-\Pr[H_1=1]
\right|
\leq
\operatorname{Adv}^{\mathsf{zk}}_{\Pi,\mathcal{B}_{\mathsf{zk},0}}
(\lambda).
\]

\noindent\emph{$(H_2)$ Simulated proof for $x_1$.}
The challenger replaces $C_0$ by $C_1$ and keeps every field in $\Framework.\mathsf{CommonCtx}$. The statement therefore becomes $x_1$. The setup and the proof remain simulated.

The adversary $\mathcal{B}_{\mathsf{hid}}$ moves from $H_1$ to $H_2$. It receives $\mathsf{pp}_{\mathsf{CS}}$ from its hiding challenger. It runs $(\widetilde{\mathsf{pp}}_{\Pi},\tau)\gets\mathsf{SimSetup}(1^\lambda)$ and the remaining setup itself, and it runs $\mathcal{A}_{2,1}$. It computes both pairs with $\Framework.\mathsf{BuildClient}$ and outputs $0$ if they are not admissible. Let $M_0$ and $M_1$ be the collections that $\Framework.\mathsf{BuildClient}$ commits to for $s_0$ and $s_1$. Both challenge executions use the same policy $\phi$. Their committed collections therefore have the same policy-defined structure, and $(M_0,M_1)$ is a valid pair for the hiding game. $\mathcal{B}_{\mathsf{hid}}$ submits $(M_0,M_1)$ and receives $C_\beta$. It forms
\[
x=(i,t,\phi,\theta_t,C_\beta,Z^{\mathsf{pub}},\Params),
\]
computes $\widetilde{\pi}\gets\mathsf{SimProve}_{\mathcal{R}_\phi}(\tau,x)$, runs $\mathcal{A}_{2,2}(\mathsf{st},x,\widetilde{\pi})$, and outputs its bit.
The commitment randomness does not affect admissibility (Section~\ref{sec:notation}), and the simulated proof uses no witness. For $\beta=0$, the view of $\mathcal{A}_2$ is therefore distributed as in $H_1$. For $\beta=1$, it is distributed as in $H_2$. By hiding of $\mathsf{CS}$,
\[
\left|
\Pr[H_1=1]-\Pr[H_2=1]
\right|
\leq
\operatorname{Adv}^{\mathsf{hid}}_{\mathsf{CS},\mathcal{B}_{\mathsf{hid}}}
(\lambda).
\]

\noindent\emph{$(H_3)$ Real proof for $x_1$.}
The challenger returns to $\ppnizk\gets\mathsf{Setup}_{\Pi}(1^\lambda)$ and replaces the simulated proof for $x_1$ by the real proof
\[
\pi_1
\gets
\Pi.\Prove_{\mathcal{R}_\phi}(\ppnizk,x_1,w_1).
\]
This is the privacy game with $b=1$.

The adversary $\mathcal{B}_{\mathsf{zk},1}$ moves from $H_2$ to $H_3$. It works as $\mathcal{B}_{\mathsf{zk},0}$, but it submits $(x_1,w_1)$ and runs $\mathcal{A}_{2,2}$ on $x_1$. In $\mathsf{ZKSim}$, it runs $H_2$. In $\mathsf{ZKReal}$, it runs $H_3$. Again by zero knowledge of $\Pi$,
\[
\left|
\Pr[H_2=1]-\Pr[H_3=1]
\right|
\leq
\operatorname{Adv}^{\mathsf{zk}}_{\Pi,\mathcal{B}_{\mathsf{zk},1}}
(\lambda).
\]

All three adversaries run $\mathcal{A}_2$ once and are therefore $\mathsf{PPT}$. Hybrid $H_0$ is the privacy game with $b=0$, and hybrid $H_3$ is the privacy game with $b=1$. The triangle inequality gives
\[
\begin{aligned}
\operatorname{Adv}^{\mathsf{priv}}_{\Framework,\mathcal{A}_2}(\lambda)
&=
\left|
\Pr[H_0=1]-\Pr[H_3=1]
\right|
\\
&\leq
\left|\Pr[H_0=1]-\Pr[H_1=1]\right| \\
&\quad+
\left|\Pr[H_1=1]-\Pr[H_2=1]\right|
\\
&\quad+
\left|\Pr[H_2=1]-\Pr[H_3=1]\right|.
\end{aligned}
\]
The three bounds above give the claim.
\end{proof}

The theorem hides only what $\Framework.\mathsf{CommonCtx}$ leaves out. We emphasise that the values in $Z^{\mathsf{pub}}$ are public by design. In particular, the theorem does not cover the confidentiality of encrypted updates (such as in the Threshold-HE realisation). That confidentiality follows from the assumptions stated in Appendix~A.3 of the paper.

\subsection{Extension to Many Messages, Many Rounds, and Committed Objects}\label{sec:ext}
Theorems~\ref{thm:update-adm} to~\ref{thm:policy-privacy} consider one round and one challenge message. Section~4 of the paper states that they extend to polynomially many clients and rounds. It also states that knowledge soundness and binding tie a client policy to the committed objects. We prove these statements here.
These statements are not new security notions. Corollary~\ref{cor:multi-sound} applies Theorems~\ref{thm:update-adm} and~\ref{thm:agg-correct} to every accepted message with a union bound. Corollary~\ref{cor:multi-priv} applies Theorem~\ref{thm:policy-privacy} to many challenge messages with a standard hybrid argument. Lemma~\ref{lem:committed} is the usual commit-and-prove argument, which combines knowledge soundness with binding~\cite{canetti2002uc,campanelli2019legosnark}.

\subsubsection{Many Rounds}
Consider an execution of $T(\lambda)$ rounds after an honest $\Framework.\mathsf{Setup}$. The adversary corrupts clients or the aggregator and acts adaptively across rounds. Honest parties follow the protocol. An honest party accepts a client message only if $\Framework.\mathsf{VfyClient}$ returns $1$ for the active context of its round. This includes the evaluations inside check (3) of $\Framework.\mathsf{VfyAgg}$. An honest party accepts an aggregation message only if $\Framework.\mathsf{VfyAgg}$ returns $1$. Let $Q(\lambda)$ bound the number of accepted client messages. We order all accepted messages by the time of acceptance. The recorded client messages of a round come before the aggregation message of that round.

\begin{corollary}[Soundness over many rounds]\label{cor:multi-sound}
Let $\mathsf{Bad}$ be the event that, in some round $t\leq T(\lambda)$, an honest party accepts a client message $\mu_{i,t}$ with $\Framework.\mathsf{ValidClient}(\kappa_t^\phi,i,\mu_{i,t})=0$ or an aggregation message $\nu_t$ with $\Framework.\mathsf{ValidAgg}(\kappa_t^\phi,\psi,\mathsf{Rec}_t,\nu_t)=0$. For every $\mathsf{PPT}$ adversary $\mathcal{A}$, there exist $\mathsf{PPT}$ adversaries $\mathcal{B}_1,\ldots,\mathcal{B}_{Q+T}$ such that
\[
\Pr[\mathsf{Bad}]
\leq
\sum_{k=1}^{Q(\lambda)+T(\lambda)}
\operatorname{Adv}^{\mathsf{snd}}_{\Pi,\mathcal{B}_{k}}(\lambda).
\]
Hence $\Pr[\mathsf{Bad}]$ is negligible when $Q$ and $T$ are polynomial and $\Pi$ is sound.
\end{corollary}

\begin{proof}
The adversary $\mathcal{B}_k$ receives $\ppnizk$ from its soundness challenger for the family $\mathcal{R}_{\Phi}$ and generates the remaining setup outputs itself. It runs the whole execution with $\mathcal{A}$ and plays all honest parties. It can do so because it knows their inputs and $\Pi.\Prove$ needs only $\ppnizk$. The contexts of later rounds depend on earlier results. They are public, and $\mathcal{B}_k$ computes them during the simulation. $\mathcal{B}_k$ outputs the statement and the proof of the $k$-th accepted message, or $\bot$ if fewer messages are accepted.

Suppose that $\mathsf{Bad}$ occurs, and let $k^\ast$ be the index of the first accepted invalid message. We consider two cases:

\emph{(i)} If it is a client message, then $\mathsf{Chk}$ succeeds and its proof verifies, while client validity fails. Its statement is therefore outside $\mathcal{L}_{\phi}$, as in the proof of Theorem~\ref{thm:update-adm}.

\emph{(ii)} If it is an aggregation message, then all recorded client messages of its round were accepted earlier and are valid. Check (2) gives distinct identifiers. Hence $\Framework.\mathsf{ValidRec}(\kappa_t^\phi,\mathsf{Rec}_t)=1$. Its statement is therefore outside $\mathcal{L}_{\psi}$, while check (4) accepts its proof. This is case (ii) of the proof of Theorem~\ref{thm:agg-correct}.
In both cases, $\mathcal{B}_{k^\ast}$ breaks soundness of $\Pi$. Therefore
\[
\Pr[\mathsf{Bad}]
\leq
\sum_{k}\Pr[k^\ast=k]
\leq
\sum_{k=1}^{Q(\lambda)+T(\lambda)}
\operatorname{Adv}^{\mathsf{snd}}_{\Pi,\mathcal{B}_{k}}(\lambda).
\qedhere
\]
\end{proof}

\subsubsection{Many Messages}
We use the multi-challenge version of the privacy game, which is the usual way to extend a single-challenge game to many queries. We define it here, since the paper states only the single-challenge game. In $\mathsf{Game}^{\mathsf{priv},n}_{\Framework,\mathcal{A}_2}(1^\lambda,b)$, the challenger runs $\Framework.\mathsf{Setup}$. It gives $\mathsf{pp}_{\Framework}$ and the prescribed setup state of the corrupted participant to $\mathcal{A}_2$. The adversary makes up to $n(\lambda)$ adaptive challenge queries. In query $k$, it declares an active public context $\kappa^{(k)}$ for a client policy $\phi$, a client $i_k$, and two private inputs $s^{(k)}_0,s^{(k)}_1$. The challenger computes $(x^{(k)}_j,w^{(k)}_j)\gets\Framework.\mathsf{BuildClient}(\kappa^{(k)},i_k,s^{(k)}_j)$ for $j\in\{0,1\}$ with fresh randomness. It answers $\bot$ if the pair is not admissible as in Fig.~\ref{fig:game-priv}. It returns $x^{(k)}_b$ and $\pi^{(k)}_b\gets\Pi.\Prove_{\mathcal{R}_\phi}(\ppnizk,x^{(k)}_b,w^{(k)}_b)$ otherwise. Finally, $\mathcal{A}_2$ outputs $b'$. The advantage $\operatorname{Adv}^{\mathsf{priv},n}_{\Framework,\mathcal{A}_2}(\lambda)$ is defined as in the single-message case. The adversary names the contexts itself, which gives it at least the power it has in a real execution. Since, there the contexts are public, and admissible pairs agree on all permitted public values.

\begin{corollary}[Privacy over many messages]\label{cor:multi-priv}
For every $\mathsf{PPT}$ verifier $\mathcal{A}_2$ that makes at most $n(\lambda)$ queries, there exist $\mathsf{PPT}$ adversaries $\mathcal{B}_{\mathsf{zk},0}$, $\mathcal{B}_{\mathsf{hid},1},\ldots,\mathcal{B}_{\mathsf{hid},n}$, and $\mathcal{B}_{\mathsf{zk},1}$ such that
\[
\begin{aligned}
\operatorname{Adv}^{\mathsf{priv},n}_{\Framework,\mathcal{A}_2}(\lambda)
\leq{}&
\operatorname{Adv}^{\mathsf{zk}}_{\Pi,\mathcal{B}_{\mathsf{zk},0}}(\lambda)
+
\sum_{k=1}^{n(\lambda)}
\operatorname{Adv}^{\mathsf{hid}}_{\mathsf{CS},\mathcal{B}_{\mathsf{hid},k}}(\lambda)
\\
&+
\operatorname{Adv}^{\mathsf{zk}}_{\Pi,\mathcal{B}_{\mathsf{zk},1}}(\lambda).
\end{aligned}
\]
\end{corollary}

\begin{proof}
Let $G_0$ be the game with $b=0$.

\noindent\underline{From $G_0$ to $G_1$.}
In $G_1$, the challenger uses $(\widetilde{\mathsf{pp}}_{\Pi},\tau)\gets\mathsf{SimSetup}(1^\lambda)$ and answers every admissible query with $x^{(k)}_0$ and a simulated proof. The adversary $\mathcal{B}_{\mathsf{zk},0}$ forwards each admissible pair $(x^{(k)}_0,w^{(k)}_0)$ to its multi-theorem zero-knowledge oracle. It runs $G_0$ in $\mathsf{ZKReal}$ and $G_1$ in $\mathsf{ZKSim}$. A single reduction switches all proofs at once. 

\noindent\underline{From $G_1$ to $G_{1,n}$.}
For $k\in\{0,\ldots,n\}$, let $G_{1,k}$ be $G_1$, except that the first $k$ admissible queries are answered with $x^{(\ell)}_1$ instead of $x^{(\ell)}_0$, again with simulated proofs. Then $G_{1,0}=G_1$. The games $G_{1,k-1}$ and $G_{1,k}$ differ only in the commitment of query $k$. The adversary $\mathcal{B}_{\mathsf{hid},k}$ runs the simulated setup itself and answers all other queries as the two games do. It embeds its hiding challenge into query $k$, as $\mathcal{B}_{\mathsf{hid}}$ does in the proof of Theorem~\ref{thm:policy-privacy}. Each query uses fresh commitment randomness, so the independence argument of Section~\ref{sec:notation} applies per query. If query $k$ does not occur or is not admissible, the two games coincide. Hence
\[
|\Pr[G_{1,k-1}=1]-\Pr[G_{1,k}=1]|\leq\operatorname{Adv}^{\mathsf{hid}}_{\mathsf{CS},\mathcal{B}_{\mathsf{hid},k}}(\lambda).
\]

\noindent\underline{From $G_{1,n}$ to $G_2$.}
Let $G_2$ be the game with $b=1$. The adversary $\mathcal{B}_{\mathsf{zk},1}$ moves from $G_{1,n}$ to $G_2$ by forwarding the pairs $(x^{(k)}_1,w^{(k)}_1)$, as $\mathcal{B}_{\mathsf{zk},0}$ does.

Finally, the triangle inequality over $G_0,G_1=G_{1,0},\ldots,G_{1,n},G_2$ gives the claimed bound.
\end{proof}

\subsubsection{Committed Objects}
Soundness shows that an accepted statement lies in $\mathcal{L}_{\phi}$. Some opening of its commitments therefore satisfies the policy. Soundness does not show that the client knows this opening. It also does not show that this is the opening used in another message with the same commitment. Such reuse occurs in ContinualFL. There, the reference-state commitment $c^{s}_{i,t-1}$ of round $t$ is the next-state commitment of the client's previous admitted message (Appendix~A.1 of the paper). Both commitments are public, so a verifier can check that they are equal. For completeness of the proofs we state here that Knowledge soundness and binding then closes this gap.

\begin{lemma}[Consistent openings of shared commitments]\label{lem:committed}
Let $\mathcal{A}$ be a $\mathsf{PPT}$ adversary that, possibly over several rounds, outputs two client messages with statements $x_1,x_2$ and proofs $\pi_1,\pi_2$. Both statements contain the same commitment $c$. Let $\mathcal{B}_1$ and $\mathcal{B}_2$ run $\mathcal{A}$ on the same inputs and output $(x_1,\pi_1)$ and $(x_2,\pi_2)$, respectively. Let $w_1\gets\mathcal{E}_{\mathcal{B}_1}$ and $w_2\gets\mathcal{E}_{\mathcal{B}_2}$ be the extracted witnesses on the same random coins, and let $(M_1,R_1)$ and $(M_2,R_2)$ be the openings of $c$ in $w_1$ and $w_2$. Let $F$ be the event that both proofs verify and that $(x_1,w_1)\notin\mathcal{R}_{\Phi}$, $(x_2,w_2)\notin\mathcal{R}_{\Phi}$, or $M_1\neq M_2$. Here, membership in $\mathcal{R}_{\Phi}$ refers to the relation in each statement. Then there exists a $\mathsf{PPT}$ adversary $\mathcal{B}_{\mathsf{bind}}$ such that
\[
\begin{aligned}
\Pr[F]
\leq{}&
\operatorname{Adv}^{\mathsf{k\text{-}snd}}_{\Pi,\mathcal{B}_1}(\lambda)
+
\operatorname{Adv}^{\mathsf{k\text{-}snd}}_{\Pi,\mathcal{B}_2}(\lambda)
\\
&+
\operatorname{Adv}^{\mathsf{bind}}_{\mathsf{CS},\mathcal{B}_{\mathsf{bind}}}(\lambda).
\end{aligned}
\]
\end{lemma}

\begin{proof}
The event $F$ is contained in the union of three events.
$F_1$ is the event that $\pi_1$ verifies and $(x_1,w_1)\notin\mathcal{R}_{\Phi}$. $F_2$ is the same event for the second message. By definition, their probabilities are at most $\operatorname{Adv}^{\mathsf{k\text{-}snd}}_{\Pi,\mathcal{B}_1}(\lambda)$ and $\operatorname{Adv}^{\mathsf{k\text{-}snd}}_{\Pi,\mathcal{B}_2}(\lambda)$. 

$F_3$ is the event that both witnesses are valid and $M_1\neq M_2$.
We use adversary $\mathcal{B}_{\mathsf{bind}}$ to bound $F_3$. It receives $\mathsf{pp}_{\mathsf{CS}}$ from its binding challenger. It passes $\mathsf{pp}_{\mathsf{CS}}$ to $\mathcal{B}_1$, $\mathcal{B}_2$, and their extractors as auxiliary input, which Section~\ref{sec:crypto} allows. It samples the random coins of $\mathcal{A}$ and runs $\mathcal{E}_{\mathcal{B}_1}$ and $\mathcal{E}_{\mathcal{B}_2}$ on these coins. It reads the openings $(M_1,R_1)$ and $(M_2,R_2)$ of $c$ from $w_1$ and $w_2$ and outputs $(M_1,R_1,M_2,R_2)$. Suppose that $F_3$ occurs. Both witnesses satisfy the commitment conditions of their policies, so
\[
c =\Commit(\mathsf{pp}_{\mathsf{CS}},M_1;R_1) =\Commit(\mathsf{pp}_{\mathsf{CS}},M_2;R_2) 
\]
such that $M_1\neq M_2$
This breaks binding. 
A union bound over $F_1$, $F_2$, and $F_3$ gives the claim.
\end{proof}

Except with negligible probability, the extracted witnesses therefore satisfy their policies, and both proofs refer to the same committed object. A client policy thus holds on the private data and state that the client committed to. The retained state of ContinualFL carries over from one admitted message to the next. The same argument applies to any number of messages that share a commitment. Each message adds one knowledge-soundness term, and each pair adds one binding term.

\section{Open Science}\label{app:open-science}
The anonymised artefact at \url{https://anonymous.4open.science/r/Verifiable-FL-Framework-E6BB/} contains the prototype, the scripts that reproduce all reported measurements, and the more detailed proofs of Appendix~\ref{app:security-proofs}.

\end{document}